\documentclass[preprint, review, 1p]{elsarticle}

\usepackage[T1]{fontenc}
\usepackage[english]{babel}
\usepackage{indentfirst}

\usepackage{amsfonts}
\usepackage{graphicx}
\usepackage[pdftex,table,dvipsnames]{xcolor}
\usepackage[colorlinks=true, allcolors=blue]{hyperref}
\usepackage{comment}
\usepackage{mathtools}
\usepackage[ruled,linesnumbered,noend]{algorithm2e}
\usepackage{array}

\usepackage{amsthm}
\usepackage{nameref}
\usepackage[capitalise]{cleveref}
\makeatletter
\newcommand{\refcheckize}[1]{%
  \expandafter\let\csname @@\string#1\endcsname#1%
  \expandafter\DeclareRobustCommand\csname relax\string#1\endcsname[1]{%
    \csname @@\string#1\endcsname{##1}\wrtusdrf{##1}}%
  \expandafter\let\expandafter#1\csname relax\string#1\endcsname
}
\makeatother

\usepackage[colorinlistoftodos,prependcaption,textsize=tiny]{todonotes}
\usepackage{xargs}
\usepackage{lineno}

\newcommand{\CTM}{\approx_{CT}}
\newcommand{\nCTM}{\not\approx_{CT}}

\newcommand{\tauCTM}{\overset{\tau}{\approx}_{CT}}
\newcommand{\misCTM}{\overset{\textsc{MIS}}{\approx}_{CT}}
\newcommand{\insCTM}{\overset{\textsc{INS}}{\approx}_{CT}}
\newcommand{\delCTM}{\overset{\textsc{DEL}}{\approx}_{CT}}

\newcommand{\PDD}{\protect\overrightarrow{PD}}
\newcommand{\PDG}{\protect\overleftarrow{PD}} 

\newcommand{\SN}{SN}
\newcommand{\SNRepres}{{\em Skipped-number} representation }

\newcommand{\LN}{LN} 
\newcommand{\LND}{\protect\overrightarrow{LN}} 
\newcommand{\LNG}{\protect\overleftarrow{LN}}  

\newcommand{\aD}{\protect\overrightarrow{a_x}}
\newcommand{\bD}{\protect\overrightarrow{b_x}}
\newcommand{\cD}{\protect\overrightarrow{a_y}}
\newcommand{\dD}{\protect\overrightarrow{b_y}}
\newcommand{\aG}{\protect\overleftarrow{a_x}}
\newcommand{\bG}{\protect\overleftarrow{b_x}}
\newcommand{\cG}{\protect\overleftarrow{a_y}}
\newcommand{\dG}{\protect\overleftarrow{b_y}}

\renewcommand{\ng}{\operatorname{NG}} 
\newcommand{\ngi}{\operatorname{ng}}  
\newcommand{\SGraph}{\mathcal{G}}     

\newcommand{\rb}{\operatorname{rb}}  
\newcommand{\rbs}{\operatorname{rbs}}  

\newcommand{\lb}{\operatorname{lb}}  
\newcommand{\RMP}{\operatorname{RB}} 
\newcommand{\LMP}{\operatorname{LB}} 

\newcommand{\leftsb}{\operatorname{left}}   
\newcommand{\rightsb}{\operatorname{right}} 

\newcommand{\REFD}[1]{\protect\operatorname{ref}_{#1}} 
\newcommand{\notREF}{RD}

\newcommand{\lcp}{\operatorname{lcp}} 
\newcommand{\lcs}{\operatorname{lcs}} 

\newcommand{\dIsEmpty}{\algoName{isEmpty}}
\newcommand{\dBack}{\algoName{back}}
\newcommand{\dPopBack}{\algoName{popBack}}

\newcommand{\dFront}{\algoName{front}}
\newcommand{\dPopFront}{\algoName{popFront}}
\newcommand{\dPushFront}{\algoName{pushFront}}

\newcommand{\algoName}[1]{\textsc{#1}}
\newcommand{\eqTest}{\algoName{equivalenceTest}} 

\newcommandx{\RQunsure}[2][1=]{\todo[linecolor=red,backgroundcolor=red!25,bordercolor=red,#1]{#2}}
\newcommandx{\RQchange}[2][1=]{\todo[linecolor=blue,backgroundcolor=blue!25,bordercolor=blue,#1]{#2}}
\newcommandx{\RQinfo}[2][1=]{\todo[linecolor=OliveGreen,backgroundcolor=OliveGreen!25,bordercolor=OliveGreen,#1]{#2}}
\newcommandx{\RQimprovement}[2][1=]{\todo[linecolor=Plum,backgroundcolor=Plum!25,bordercolor=Plum,#1]{#2}}

\newdefinition{definition}{Definition}
\newdefinition{example}{Example}

\newtheorem{thm}{Theorem}
\newtheorem{lemma}[thm]{Lemma}
\newtheorem{namedlemma}[thm]{Lemma}
\newtheorem{remark}[thm]{Remark}
\newtheorem{coro}[thm]{Corollary}

\crefname{lemma}{lemma}{lemmas}
\Crefname{lemma}{Lemma}{Lemmas}
\crefname{thm}{theorem}{theorems}
\Crefname{thm}{Theorem}{Theorems}

\begin{document}

\title{Approximate Cartesian Tree Matching with One Difference\tnoteref{t1}}
\tnotetext[t1]{A previous version of this paper was presented at the 30th edition of the Symposium on String Processing and Information Retrieval (SPIRE 2023).}

\author[1,3]{Bastien Auvray\texorpdfstring{\corref{cor1}}}
\ead{bastien.auvray@univ-rouen.fr}

\author[2,3]{Julien David}
\ead{julien.david@unicaen.fr}

\author[4]{Samah Ghazawi}
\ead{samahi@braude.ac.il}

\author[1,3]{Richard Groult}
\ead{richard.groult@univ-rouen.fr}

\author[5,6]{Gad M. Landau}
\ead{landau@univ.haifa.ac.il}

\author[1,3]{Thierry Lecroq}
\ead{thierry.lecroq@univ-rouen.fr}

\cortext[cor1]{Corresponding author}

\affiliation[1]{organization={Univ Rouen Normandie, INSA Rouen Normandie, Université Le Havre Normandie, Normandie Univ, LITIS UR 4108},
city={Rouen},
postcode={76000},
country={France}}

\affiliation[2]{organization={Normandie University, UNICAEN, ENSICAEN, CNRS, GREYC},
city={Caen},
country={France}}

\affiliation[3]{organization={CNRS NormaSTIC FR 3638},
city={Caen, Le Havre, Rouen},
country={France}}

\affiliation[4]{organization={Department of Software Engineering, Braude, College of Engineering}, city={Karmiel}, country={Israel}}

\affiliation[5]{organization={Department of Computer Science, University of Haifa}, city={Haifa}, country={Israel}}

\affiliation[6]{organization={Department of Computer Science and Engineering, NYU Tandon}, city={New York}, country={USA}}

\begin{abstract}
Cartesian tree pattern matching consists of finding all the factors of a text that have
 the same Cartesian tree than a given pattern.
There already exist theoretical and practical solutions for the exact case.
In this paper, we propose the first algorithms for solving approximate Cartesian tree
 pattern matching with one difference given a pattern of length $m$
 and a text of length $n$.
We present a generic algorithm that find all the factors of the text that have the same Cartesian tree of
 the pattern with one difference, using different notions of differences.
 We show that this algorithm has a $\mathcal{O}(nm)$ worst-case complexity and that, 
 for several random models, the algorithm has a linear average-case complexity.
We also present an automaton based algorithm, adapting~\cite{PALP19}, that can be generalized
to deal with more than one difference.
\end{abstract}


\maketitle
\newpage

\tableofcontents


\section{Introduction}

In general terms, the pattern matching problem consists of finding one or all
 the occurrences of a pattern $p$ of length $m$ in a text $t$ of length $n$.
When both the pattern and the text are strings the problem has been extensively
 studied and has received a huge number of solutions~\cite{FL2013}.
Searching time series or list of values for patterns representing specific
 fluctuations of the values requires a redefinition of the notion of pattern.
The question is to deal with the recognition of peaks, breakdowns, or more features.
For those specific needs one can use the notion of Cartesian tree. 

Cartesian trees have been introduced by Vuillemin in 1980~\cite{vuillemin80}.
They are mainly associated to strings of numbers and are structured as heaps
 from which original strings can be recovered by symmetrical traversals of the trees.
It has been shown that they are connected to Lyndon trees~\cite{CROCHEMORE20201},
 to Range Minimum Queries~\cite{demaine09}
 or to parallel suffix tree construction~\cite{SB14}.
Recently, Park \textit{et al.}~\cite{PALP19} introduced a new metric of generalized matching,
 called Cartesian tree matching.
It is the problem of finding every factor of a text $t$ which has the same
 Cartesian tree as that of a given pattern $p$.
Cartesian tree matching can be applied, for instance, to finding patterns in time series
 such as share prices in stock markets or gene sample time data.
 
Park \textit{et al.} introduced the parent-distance representation which is a
 linear form of the Cartesian tree and that has a one-to-one mapping with Cartesian trees.
They gave linear-time solutions for single and multiple pattern Cartesian tree matching,
 utilizing this parent-distance representation and existing classical string algorithms,
 i.e., Knuth-Morris-Pratt~\cite{KMP77} and Aho-Corasick~\cite{AC75} algorithms.
More efficient solutions for practical cases were given in~\cite{SGRFLP21}.
Recently, new results on Cartesian pattern matching appeared~\cite{PARK2020,FLPS22}.

Indexing structures in the Cartesian tree pattern matching 
 framework are presented in~\cite{NFNI2021,KimC21,osterkamp2024}.
Methods for computing regularities are given in~\cite{KikuchiHYS20}
 and methods for computing palindromic structures are presented
 in~\cite{Funakoshi_et_al}.

All these previous works on Cartesian tree matching are concerned with
 finding exact occurrences
 of patterns consisting of contiguous symbols.
An algorithm for episode matching
 (finding all minimal length factors of $t$ that contain $p$ as a subsequence) in Cartesian
 tree framework is presented in~\cite{OizumiKMIA22}.

Very recently, dynamic programming approaches for approximate Cartesian
 tree pattern matching with edit distance has been
 considered in~\cite{KimH24} and longest common Cartesian tree
 subsequences are computed in~\cite{TsujimotoSMNI24}.
 
Efficient algorithms for determining if two equal-length indeterminate
 strings match in the Cartesian tree framework are given in~\cite{gawrychowski_et_al:LIPIcs.CPM.2020.14}.

In real life applications data are often noisy and it is thus important to
 find factors of the text that are similar, to some extent, to the pattern.
In this paper, we present a set of results in this setting by considering
 approximate Cartesian tree pattern matching with one difference, be it either a transposition (aka swap) of one symbol with the adjacent symbol, a mismatch, an insertion of one symbol or a deletion of one symbol. In a previous version of this paper~\cite{10.1007/978-3-031-43980-3_5}, we gave preliminary results for approximate Cartesian tree matching with one swap.
Swap pattern matching has received a lot of attention in classical sequences
 since the first paper in 1997~\cite{AmirALLL97}
 (see~\cite{FaroP18} and references therein).
Swaps are common in real life data and it seems natural to consider them in the
 Cartesian pattern matching framework.
We are able to design two algorithms for solving the Cartesian tree pattern matching
 with at most one swap.
The first one runs in time $\Theta(mn)$ and uses a characterization of a linear representation
 of Cartesian trees while the second one runs in $\mathcal{O}((m^2 + n)\log{m} )$ and uses an automaton that recognizes
 all the linear representations of Cartesian trees of sequences that match the pattern after one swap, and we show that the size of the automaton is bounded by $3(m-2)+1$. We also present methods to considerate approximate Cartesian tree matching with one mismatch, one insertion or one deletion.

The remaining of the article is organized as follows: 
\Cref{sec:pre}
 presents the basic notions and notations used in the paper. It also presents 
 two linear representations of Cartesian trees respectively called the parent-distance and the Skipped-Number representations. 
 \Cref{sec:approx} describes different notions
 of approximate Cartesian tree matching up to one error, based on: swaps, insertions, mismatches and deletions. We show that there exists a generic algorithm to solve those problems and study its time complexity. The generic algorithm uses a function that
 tests whether two Cartesian trees are equivalent, called \algoName{\eqTest}, that depends on both
 the chosen linear representation and the chosen notion of approximate matching.
 The following Sections therefore contain studies of the consequences of a ``mistake''
 on the linear representations and describe the \algoName{\eqTest} function.
Given a sequence $x$, 
 in \Cref{sec:charac} we give a characterization of the parent-distance representations
 of Cartesian trees that correspond to sequences $x$ after one swap.
 \Cref{sec:skipped} does the same for the Skipped-Number representation.
 \Cref{sec:mismatch,sec:insert,sec:delete} respectively inspect the effect
 of a mismatch, an insertion and a deletion on both linear representations.
\Cref{sec:graph} refocuses on the notion of swap. 
A graph is introduced, where vertices are Cartesian trees
 and there is an edge between two vertices if both Cartesian trees can be obtained from the other using one swap. Another algorithm to solve approximate pattern matching is obtained from it.
In \Cref{sec:expe} we give experimental results.
\Cref{sec:persp} contains our perspectives.

\section{Preliminaries \label{sec:pre}}
\subsection{Basic notations}

We consider sequences of integers with a total order denoted by $<$. 
For a given sequence $x$, $\vert x\vert$ denotes the length of $x$.
A sequence $v$ is factor of a sequence $x$ if $x=uvw$ for any sequences $u$ and $w$.
A sequence $u$ is a prefix (resp. suffix) of a sequence $x$ if $x=uv$ (resp. $x=vu$).
For a sequence $x$ of length $m$,
 $x[j]$ is the $j$-th element of $x$ and $x[j\ldots k]$ represents the factor of $x$ starting at the $j$-th element and ending at the $k$-th element, for $1\le j\le k\le m$.
For simplicity, we assume all elements in a given sequence to be distinct and numbered from 1 to $\vert x\vert$, unless otherwise stated.

\subsection{Cartesian tree matching}

\begin{definition}[Cartesian tree $C(x)$, $C_h(x)$, $\mathcal{C}$, $\mathcal{C}_m$] \label{def:ct}
    Given a sequence $x$ of length $m$, its \emph{Cartesian tree}, denoted by $C(x)$, is the binary tree recursively defined as follows:
    \begin{itemize}
        \item if $x$ is empty, then $C(x)$ is the empty tree;
        \item if $x[1\ldots m]$ is not empty and $x[g]$ is the smallest value of $x$, $C(x)$ is the binary tree with a node labeled by $g$ (called node $g$ for short) as its root, the Cartesian tree of $x[1\ldots g-1]$ as the left subtree and the Cartesian tree of $x[g+1\ldots m]$ as the right subtree.
    \end{itemize}
    
    Let $C_j(x)$ denote the subtree rooted at node $h$ of $C(x)$ the Cartesian tree of a given sequence $x$, where $1\le h \le m$.
    
    Let $\mathcal{C}_m$ be the set of Cartesian trees with $m$ nodes, which is equal to the set of binary trees with $m$ nodes. 
    Also $\mathcal{C} = \bigcup_{m \ge 0} \mathcal{C}_m$ denotes the set of all Cartesian trees. 
\end{definition}

See \Cref{fig:parent} for an example.

\begin{figure}[!ht]
\begin{center}
\includegraphics[scale=0.8]{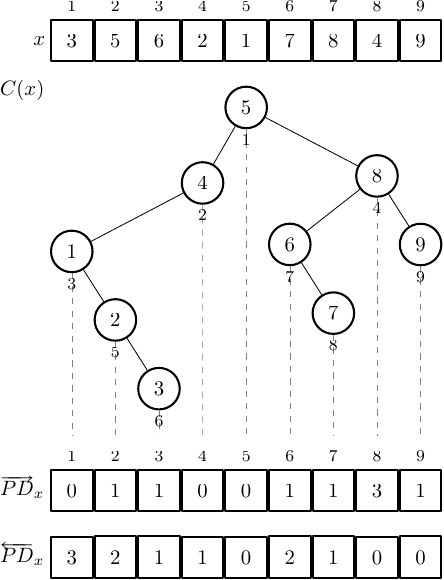}
\caption{
A sequence $x = (3, 5, 6, 2, 1, 7, 8, 4, 9) $, its Cartesian tree $C(x)$ and its corresponding parent-distance table $\PDD_x$ and $\PDG_x$. 
We have $\rb(C(x)) = (5,8,9)$.
\label{fig:parent}
}
\end{center}
\end{figure}

\begin{definition}[left and right subtrees and branches]
Let $T$ be a non-empty binary tree.

We denote by $\leftsb(T)$ and $\rightsb(T)$  its \emph{left and right subtrees}.
We denote by $\lb(T)$ and $\rb(T)$ the list of nodes on the \emph{left and right branches} of $T$
 respectively. 

Let $\LMP: \mathcal{C} \mapsto \mathbb{N}$ be the length of the left-branch (or leftmost path) of $T$:
$\LMP(T) = \vert \lb(T) \vert $.
The equivalent $\RMP$ function is the length of the right-branch (or rightmost path): 
$\RMP(T) = \vert \rb(T) \vert $.
\end{definition}

The Cartesian tree of a sequence can be built online in linear time and space~\cite{CTtime}.
Informally,
 let $x$ be a sequence such that $C(x[1\ldots h-1])$ is already known. 
 In order to build $C(x[1\ldots h])$, one only needs to find the nodes $j_1<\cdots<j_k$ in $\rb(C(x[1\ldots h-1]))$ such that $x[j_1] > x[h]$ (see \Cref{fig:insere_sommet}).
If $j_1$ is the root of $C(x[1\ldots h-1])$ then $h$ becomes the root of $C(x[1\ldots h])$
 otherwise let $j_0$ be the parent of $j_1$, then node $h$ will be the root of the right subtree of $j_0$.
In both cases, $j_1$ will be the root of the left subtree of node $h$.
Then $\rb(C(x[1\ldots h])=\rb(C(x[1\ldots h-1]))\setminus(j_1,\ldots,j_k) \cup (h)$.
All these operations can be easily done by implementing $\rb$ with a stack.
Each element of the stack consists of  a pair $(val,pos)$ where $val$ is a symbol of $x$ and $pos$ is the
 associated position.
The amortized cost
 of such an operation can be shown to be constant.
 
\begin{figure}[!ht]
    \begin{center}
        \includegraphics{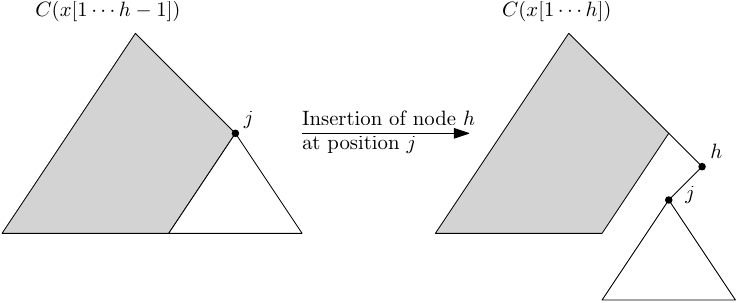}
        \caption{The construction of $C(x[1\ldots h])$ from $C(x[1\ldots h-1])$. Node $h$ is inserted at the end of the rightmost path of $C(x[1\ldots h-1])$.
        \label{fig:insere_sommet}}
    \end{center}
\end{figure}

We will denote by $x \CTM y$ if sequences $x$ and $y$ share the same Cartesian tree. 
For example, $(3, 5, 6, 2, 1, 7, 8, 4, 9) \CTM (3, 4, 8, 2, 1, 7, 9, 5, 6)$.

The Cartesian tree matching (CTM) problem consists of finding all factors of a text which share the same Cartesian tree as a pattern. 
Formally, Park \textit{et al.}~\cite{PALP19} define it as follows:

\begin{definition}[Cartesian tree matching (CTM)] \label{ctm}
Given two sequences $p[1\ldots m]$ and $t[1\ldots n]$, find every $1 \le j \le n - m + 1$ such that $t[j\ldots j+m-1] \CTM p[1\ldots m]$.
\end{definition}

\subsection{Linear representations \label{sec:lineardef}}

We will now present two linear representations of Cartesian trees
 that are used for solving exact and approximate Cartesian tree matching problems.

\subsubsection{Parent-distance and reverse parent-distance}

In order to solve CTM without building every possible Cartesian tree, an efficient representation of these trees was introduced by Park \textit{et al.}~\cite{PALP19}, the parent-distance representation (see example \figurename~\ref{fig:parent}):

\begin{definition}[Parent-distance representation $\PDD_x$] \label{def:pd}
Given a sequence $x[1\ldots m]$, the \emph{parent-distance representation} of $x$ is an integer sequence $\PDD_x[1\ldots m]$, which is defined as follows:
\begin{align*}
\PDD_x[h] = 
\begin{cases*}
h - \max_{1\leq j<h}\{j\ |\ x[j] < x[h]\}& \mbox{if such $j$ exists}\cr
0 & \mbox{otherwise.}
\end{cases*}
\end{align*}
\end{definition}

Since the parent-distance representation has a one-to-one mapping with Cartesian trees, it can replace them without loss of information.

Next, in order to fully characterize the approximate Cartesian tree matching problems that will be defined later in \Cref{sec:approx}, we introduce the notion of reverse parent-distance of a sequence that we compute as if read from right to left (see example \Cref{fig:parent}).

\begin{definition}[Reverse parent-distance representation $\PDG_x$] \label{rev-pd}
Given a sequence\\ $x[1\ldots m]$, the reverse parent-distance representation of $x$ is an integer sequence\\ $\PDG_x[1\ldots m]$, which is defined as follows:
$$\PDG_x[h] = \begin{cases}
\min_{h < j \leq m}\{j\ |\ x[h] > x[j]\} - h & \mbox{if such $j$ exists}\\
0 & \mbox{otherwise.}
\end{cases}
$$
\end{definition}

\subsubsection{Skipped-number representation}

The parent-distance table is not the only linear representation
 of Cartesian trees.
We will now present another linear representation based on the number
 of nodes skipped by a new node when building online the Cartesian tree
 of a sequence. Informally, for each node $h$, we store the nodes $(j_1 < \cdots < j_k)$ deleted from the right path when computing $C_h(x)$ from $C_{h-1}(x)$. We say that node $h$ skipped nodes $(j_1 < \cdots < j_k)$ and that these nodes are skipped by node $h$.  

\begin{definition}[Skipped-nodes representation $\rbs_x$] \label{def:skippednodes}
Given a sequence $x[1\ldots m]$, the {\em Skipped-nodes representation} of $x$ is a sequence of sets $\rbs_x[1\ldots m]$
such that $\rbs_x[h]$ is the right-branch of the left-subtree of the subtree rooted at node $h$ of the Cartesian tree of $x$, that is
$\rbs_x[h] = \rb( \leftsb(C_h(x))). $
\end{definition}

\begin{definition}[Skipped-number representation $\SN_x$] \label{def:SN}
Given a sequence $x[1\ldots m]$, the {\em Skipped-number representation} of $x$ is an integer sequence $\SN_x[1\ldots m]$
such that $\SN_x[h]$ is the length of the right-branch of the left-subtree of the subtree rooted at node $h$ of the Cartesian tree of $x$, that is
$\SN_x[h] = \RMP( \leftsb(C_h(x)) ) = |\rbs_x[h]|.$
\end{definition}

\begin{figure}[!ht]
    \centering
    \includegraphics[scale=0.40]{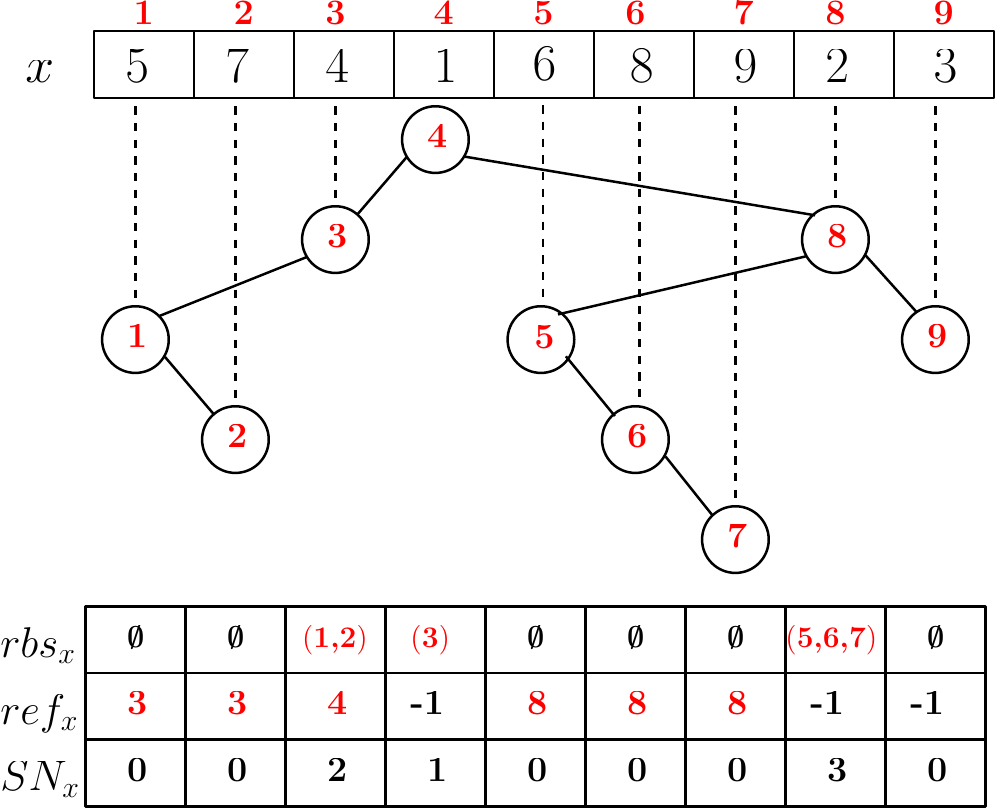}
    \caption{\label{fig:C-X} The Cartesian trees of $x=(5,7,4,1,6,8,9,2,3)$ and its associated tables.}
\end{figure}

We note that a similar notion to the Skipped-number representation appeared in Ohlebusch's book~\cite{Ohlebusch} in Chapter 3 (see also~\cite{FischerH06}). This notion was also present in~\cite{demaine09} and~\cite{PALP19} albeit under a different name: the Cartesian tree signature.

The Skipped-number representation also has a one-to-one mapping with Cartesian trees.

We will call referent of node $j$ the node $h$ that skipped node $j$ during the online construction of the Cartesian tree.

\begin{definition}[Referent $\REFD{x}$]\label{def:ref}
Given a sequence $x$ of length $m$,
let $\REFD{x}: \{1,\ldots,m\} \mapsto \{2,\ldots,m\}\cup \{-1\}$ 
be a function that associates to each position $h$ 
the smallest position $j>h$ such that $x[j] < x[h]$.
When  $\REFD{x}(h) = j$, the position $j$ is called the {\em referent} of $h$. 

\end{definition}
See \figurename~\ref{fig:C-X} for an example.
We remark that for all position $h$, $\REFD{x}(h) = j$ if and only if $h \in \rb(\leftsb(C_j(x))$. 
Note that, $\REFD{x}(h)\neq 1$ in any case, since the first position is added to an empty tree.

\section{Approximate Cartesian tree matching \label{sec:approx}}
In this section, we define several kinds of approximate Cartesian tree matching notions,
always up to one difference.
Then, we exhibit a generic algorithm to solve those problems and study its best-case, 
worst-case and average-case complexity under reasonable assumptions.

\subsection{One swap}

In order to define an approximate version of Cartesian tree matching, we start by considering the following notion of transposition on sequences: 

\begin{definition}[Swap $\tau(x,i)$]\label{def:swap}
Let $x$ be a sequence of length $m$, and $i \in \{1,\ldots, m-1\}$, we denote $y = \tau(x,i)$ the sequence obtained by a \emph{swap} at position $i$ in $x$, that is:
$$y = \tau(x,i) \text{ if }\begin{cases}
y[j] = x[j], \forall j \notin \{i,i+1\}\\
y[i] = x[i+1] \\
y[i+1] = x[i]
\end{cases}
$$
\end{definition}

This kind of transposition is the one made by the Bubble Sort algorithm
and the one appearing in the Damerau-Levenshtein distance.
It is therefore a natural operation on permutations and sequences.
For the Cartesian tree point of view, see~\figurename~\ref{fig:transposition}.
We use the notion of swap to define an instance of approximate Cartesian tree matching.

\begin{definition}[$CT_\tau$ matching] \label{ct-tau}
Let $x$ and $y$ be two sequences of length $m$, we say that $x\ CT_\tau$ matches $y$ (denoted $x \tauCTM y$) if:

$$ \begin{cases} x \CTM y \text{, or}\\ 
\exists\ x',\ y', \exists\ i\in\{1,\ldots, m-1\}, x' \CTM x, y' \CTM y,  x' = \tau(y', i)\ and\ y' = \tau(x', i)\end{cases}$$
\end{definition}

\begin{example}
    $(2,3,4,1,5,7,8,6,9) \tauCTM (4,5,6,3,1,7,8,2,9)$, 
    see \Cref{fig:transposition}.
\end{example}

\begin{figure}
    \centering
    \includegraphics[scale=0.66]{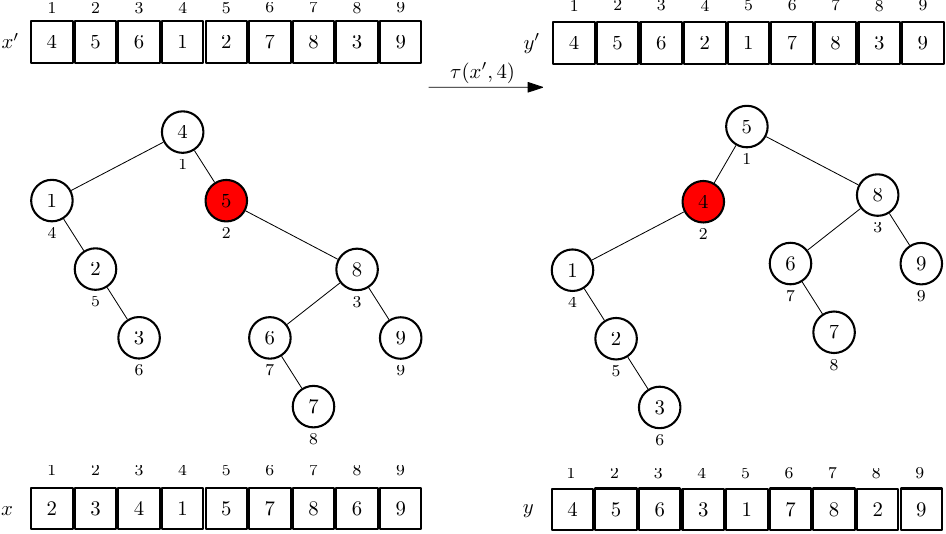}
    \caption{Let $x=(2,3,4,1,5,7,8,6,9,)$ and $y=(4,5,6,3,1,7,8,2,9)$. %
    The sequence $x$ $CT_\tau$ matches $y$ (see Definition \ref{ct-tau} with $x'=(4,5,6,1,2,7,8,3,9)$ and $y'=(4,5,6,2,1,7,8,3,9)$). 
    A swap at position $4$ moves the red node from the right subtree of the root to the left one. 
    In general, a swap at position $i$ consists either in moving
    the leftmost descendant of the right subtree to \textbf{a} rightmost position in 
    the left subtree (that is if $x[i] < x[i+1]$), or the opposite, in moving
    the rightmost descendant of the left subtree to
    \textbf{a} leftmost position of the right subtree of its parent.
    Note that we also have $x \tauCTM y'$, $x' \tauCTM y$ and of course $x' \tauCTM y'$.}
    \label{fig:transposition}
\end{figure}

With that in mind, we now define the version of the approximate Cartesian tree matching problem with at most one swap.

\begin{definition}[Approximate Cartesian tree matching with one swap]\label{ACTM}
Gi\-ven two sequences $p[1\ldots m]$ and $t[1\ldots n]$, find every $1 \le j \le n - m + 1$ such that $t[j\ldots j+m-1] \tauCTM p[1\ldots m]$.
\end{definition}



\subsection{One mismatch}
Informally, two sequences of equal length $m$ match
 with up to one mismatch if they have a prefix of length $k$
 that have the same Cartesian tree and if they have a suffix of length
 $m-k-1$ that have the same Cartesian tree.

\begin{definition}[$CT_\textsc{mis}$ matching] \label{ct-mis}
Let $x$ and $y$ be two sequences of length $m$, we denote $x \misCTM y$ if 
$\exists\ h\in\{1,\ldots, m\}$ such that:
$$ \begin{cases}
x[1\ldots h-1] \CTM y[1\ldots h-1]\\
x[h+1\ldots m] \CTM y[h+1\ldots m].
\end{cases}$$
\end{definition}

\begin{example}
    $(2,3,4,1,5,7,8,6,9) \misCTM (3,4,9,2,5,6,8,1,7)$ since
    $(2,3,4,1) \CTM (3,4,9,2)$ and
    $(7,8,6,9) \CTM (6,8,1,7)$ (see Figure \ref{fig:CT-MIS-INS-DEL} ($a$)).
\end{example}

\begin{figure}
\begin{minipage}{0.28\textwidth}
    \includegraphics[scale=0.6]{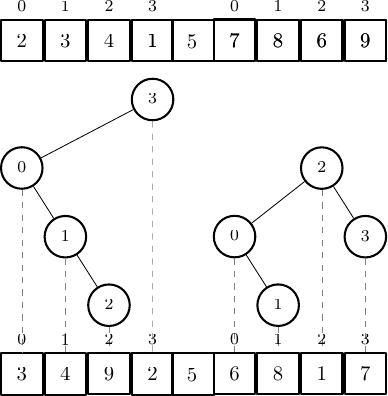}
    $$(a)~CT_\textsc{mis} \text{ matching}$$
\end{minipage}
\hspace{0.5em}\vrule\hspace{0.5em}
\begin{minipage}{0.33 \textwidth}
     \includegraphics[scale=0.6]{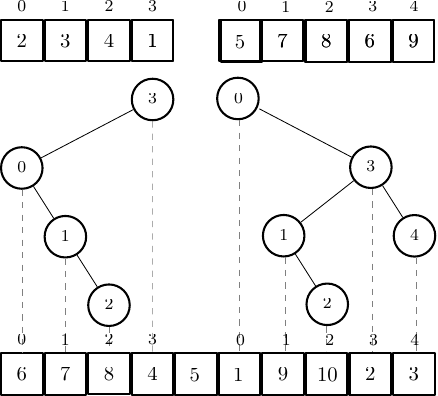}
    $$(b)~CT_\textsc{ins}  \text{ matching}$$
\end{minipage}
\hspace{0.5em}\vrule\hspace{0.5em}
\begin{minipage}{0.28\textwidth}
    \includegraphics[scale=0.6]{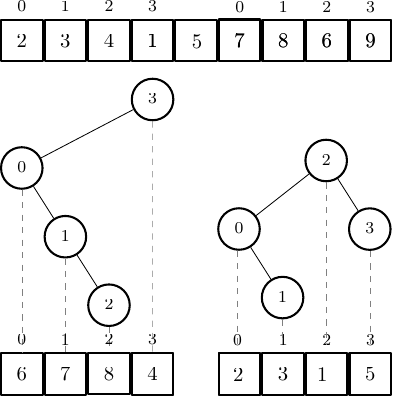}
    $$(c)~CT_\textsc{del} \text{ matching}$$
\end{minipage}

    \caption{Examples of approximate Cartesian tree matching when the error
    respectively comes, from left to right, from a mismatch, an insertion or a deletion. The sequence at the top does not change, but the equivalent one at the bottom does.}
    \label{fig:CT-MIS-INS-DEL}
\end{figure}

The Cartesian tree pattern matching with one mismatch consists
 of finding all the factors of a text $t$ that Cartesian tree
 match with up to one mismatch with a pattern $p$.

\begin{definition}[Approximate Cartesian tree matching with one mismatch]\label{MCTM}
Gi\-ven two sequences $p[1\ldots m]$ and $t[1\ldots n]$, find every $1 \le j \le n - m + 1$ such that $p[1\ldots m] \misCTM  t[j\ldots j+m-1] $.
\end{definition}

\subsection{One insertion}

Informally, one sequence of length $m$ matches
 with one insertion with a sequence of length $m+1$ if they have a prefix of length $k$
 that have the same Cartesian tree and if they have a suffix of length
 $m-k$ that have the same Cartesian tree.

\begin{definition}[$CT_\textsc{ins}$ matching] \label{ct-ins}
Let $x$ and $y$ be two sequences of length $m$ and $m+1$ respectively,
 we denote $x \insCTM y$ if $\exists\ h\in\{1,\ldots, m\}$ such that:
$$ \begin{cases} 
x[1\ldots h] \CTM y[1\ldots h]\\
x[h+1\ldots m] \CTM y[h+2\ldots m+1].
\end{cases}$$
\end{definition}

\begin{example}
    $(2,3,4,1,5,7,8,6,9) \insCTM (6,7,8,4,5,1,9,10,2,3)$ since
    $(2,3,4,1) \CTM (6,7,8,4)$ and
    $(5,7,8,6,9) \CTM (1,9,10,2,3)$ 
    (see Figure \ref{fig:CT-MIS-INS-DEL} ($b$)).
\end{example}


The Cartesian tree pattern matching with one insertion consists
 of finding all the factors of a text $t$ that Cartesian tree
 match with one insertion with a pattern $p$.

\begin{definition}[Approximate Cartesian tree matching with one insertion]\label{ICTM}
Gi\-ven two sequences $p[1\ldots m]$ and $t[1\ldots n]$, find every $1 \le j \le n - m$ such that $p[1\ldots m] \insCTM t[j\ldots j+m] $.
\end{definition}

\subsection{One deletion}

Informally, one sequence of length $m$ matches
 with one deletion with a sequence of length $m-1$ if they have a prefix of length $k$
 that have the same Cartesian tree and if they have a suffix of length
 $m-k-1$ that have the same Cartesian tree.

\begin{definition}[$CT_\textsc{del}$ matching] \label{ct-del}
Let $x$ and $y$ be two sequences of length $m$ and $m-1$ respectively,
 we denote $x \delCTM y$ if $\exists\ h\in\{1,\ldots,m\}$ such that:
\[
\begin{cases} 
x[1\ldots h] \CTM y[1\ldots h]\\
x[h+2\ldots m] \CTM y[h+1\ldots m-1].
\end{cases}
\]
\end{definition}

\begin{example}
    $(2,3,4,1,5,7,8,6,9) \delCTM (6,7,8,4,2,3,1,5)$ since
    $(2,3,4,1) \CTM (6,7,8,4)$ and
    $(7,8,6,9) \CTM (2,3,1,5)$
    (see Figure \ref{fig:CT-MIS-INS-DEL} ($c$)).
\end{example}


The Cartesian tree pattern matching with one deletion consists
 of finding all the factors of a text $t$ that Cartesian tree
 match with one deletion with a pattern $p$.

\begin{definition}[Approximate Cartesian tree matching with one deletion]\label{DCTM}
Gi\-ven two sequences $p[1\ldots m]$ and $t[1\ldots n]$, find every $1 \le j \le n - m+2$ such that $p[1\ldots m] \delCTM t[j\ldots j+m-2]$.
\end{definition}

\subsection{MetaAlgorithm}

Since we have introduced several linear representations of Cartesian trees, 
we introduce a generic notation that will be used to describe a generic algorithm.

\begin{definition}[Linear representations $\LN_x$]\label{def:ln}
Given a sequence $x$
and its associated Cartesian tree $C(x)$, let $\LN_x$, or $\LN(C(x))$, 
be an ordered pair $(\LND_x, \LNG_x)$, where $\LND_x$ is a linear representations 
of $C(x)$ (either $\PDD_x$ or $\SN_x$) and $\LNG_x$ is another linear representation of $C(x)$
(either $\PDG_x$ or $\rbs_x$).
\end{definition}


\begin{algorithm}[ht!]\label{algo:meta}
  \caption{\algoName{metaAlgorithm}($p$, $t$)}
  \SetAlgoLined
  \SetKwInOut{KwIn}{Input}
  \SetKwInOut{KwOut}{Output}
  \KwIn{Two sequences $p$ and $t$ of length $m$ and $n$}
  \KwOut{The number of positions $j$ such that $p$ is equivalent to $t[j\ldots j+m-1]$ }
  $occ \leftarrow 0$\\
  $x \leftarrow t[1\ldots m]$\\
  $\LN_{p}, \LN_{x} \leftarrow$ linear representations of $C(p)$ and $C(x)$\\
  \For{$j \in \{1,\ldots,n-m+1\}$}{
    \If{$\algoName{\eqTest}(\LN_{p}, \LN_{x})$}{
        $occ \leftarrow occ +1$
    }
    $x \leftarrow t[j+1\ldots j+m]$\\
    Update $\LN_{x}$\\
  }
  \Return{$occ$}
\end{algorithm}

Assume we have a function \algoName{\eqTest}
which is adapted according to the notion of (approximate) pattern matching, 
that takes the linear representations of two sequences $x$ and $p$, returns whether
$x$ is equivalent to $p$ or not.
The \algoName{metaAlgorithm} (see~Algorithm~\ref{algo:meta}) returns the number of occurrences of the
pattern $p$ in the sequence $t$.
It can be assumed that, those representations can be computed in linear time and the update (Line 8 of Algorithm~\ref{algo:meta}) can be made in amortized constant time 
(as shown in Algorithms~\ref{algo:updatePD} and~\ref{algo:updateSN}). 
We also assume that, in the worst-case, the \algoName{\eqTest}
function performs a comparison of both linear representations until a mismatch
occurs in both directions (or everything matches), 
plus a constant number of 
comparisons in order to check the equivalence. Thus, the number of comparisons
is bounded by $m+c$, where $m$ is the length of the pattern and $c$ is a constant.

\begin{remark}
Assuming \algoName{\eqTest} has a linear worst-case complexity and a constant best-case complexity, then the \algoName{MetaAlgorithm} has a $\Theta(mn)$ worst-case time complexity and a $\Theta(n)$ best-case complexity, where $m$ is the length of $p$ and $n$ the length of $t$.
\end{remark}

A natural question to be asked when the worst-case complexity of an algorithm
differs from the best-case complexity is the behaviour of the algorithm in practice, in various contexts. A possible way to answer this question from a theoretical point of view is to consider the average-case complexity under various random models. The following lemma describes a set of random models for which
the \algoName{MetaAlgorithm} has an average-case behaviour close to its best-case complexity.

\begin{namedlemma}[Kappa]\label{lm:kappa}
Let us consider an $\eqTest$ function between linear representations of two sequences $x$ and $p$ of length $m$.
 Assuming we have a probabilistic model that guarantees that there exists a constant $\kappa \in (0,1)$ such that for all position $1\le i \le m-1$ we have: 
$$\begin{cases} \mathbb{P}(\LND_{x}[1\ldots i] = \LND_{p}[1\ldots i]\ |\ \LND_{x}[1\ldots i -1] = \LND_{p}[1\ldots i -1]) \leq \kappa \\ \mathbb{P}(\LNG_{x}[m-i\ldots m] = \LNG_{p}[m-i\ldots m]\ |\ \LNG_{x}[m-i+1\ldots m] = \LNG_{p}[m-i+1\ldots m]) \leq \kappa \end{cases}$$
 then the average-case complexity of $\eqTest$ is $\Theta(1)$.
\end{namedlemma}

\begin{proof}
  The assumption of the lemma tells us that, whether we read the linear representation from left to right 
  or from right to left, the probability that the \algoName{\eqTest} function performs more than $i$ comparisons
  between two linear representations is less than $1/\kappa^{i}$.
  As a matter of fact, it gives us an upper bound on the average number of comparisons, 
  that follows a geometric law of parameter $(1-\kappa)$, which 
  implies 
  an average cost of $\Theta(1)$. 
\end{proof}
Note that the Lemma Kappa holds for some memoryless sources and some Markovian sources. Also, note that, as we will show in Section~\ref{sec:skipped}, the \algoName{\eqTest} using the \SNRepres does not need
a reverse representation $\LNG_{x}$. Therefore the assumption of Lemma Kappa 
on $\LNG_{x}$ is not needed in this case.
Assuming we have $\LND_{x}[1\ldots j-1] = \LND_{p}[1\ldots j-1]$, then there exists a Cartesian tree $A$ whose linear representation is $\LND_{x}[1\ldots j-1]$ with a right branch of length $\RMP(A)$. 
As shown in \Cref{fig:insere_sommet}:
assuming we have already determined that $C(x[1\ldots j-1])=C(p[1\ldots j-1])$, 
to determine whether $C(x[1\ldots j])=C(p[1\ldots j])$ is equivalent to test whether the node $j$ is inserted at the same position in the right branch in both $C(p[1\ldots j-1])$ 
and $C(x[1\ldots j-1])$.

\begin{lemma}\label{lm:avbin}
For a fixed Cartesian tree pattern of size $m$, assuming the linear representation of the text is drawn over the uniform distribution over binary trees of size $n$, then Lemma~\nameref{lm:kappa} holds for $\kappa = 1/2$.
\end{lemma}
\begin{proof}
Note that the uniform distribution over binary trees of size $n$ implies that the distribution over
the linear representation of any factor $x$ of length $m$ is the uniform distribution
over binary trees of length $m$.
Let $A$ 
be the tree of size $h-1$ such that $\LND(A)=\LND_{x}[1\ldots h-1] = \LND_{p}[1\ldots h-1]$, there exists exactly $\RMP(A) +1$ distinct trees $B$ of size $h$ such that $\LND(A) = \LND(B)[1\ldots h-1]$. 
Since each tree of size $h$ is drawn with equal probability, and $\RMP(A) +1 \ge 2$, the announced result holds. 
The same logic applies if $\LNG(A)=\LNG_{x}[m-h+1\ldots m] = \LNG_{p}[m-h+1\ldots m]$.
\end{proof}

We now wish to study the average complexity of the \algoName{MetaAlgorithm} when the uniform distribution over
permutations is considered for the text. TO do so, we need to study the number of permutations associated
to a same Cartesian tree.

Given a binary tree $A$ over $n$ nodes, 
the number of permutations, whose associated Cartesian tree is $A$, is given by 
the recursive formula:
$$p(A) = \binom{n-1}{|B|} p(B) p(C)$$
where $B=left(A)$ and $C=right(A)$. Indeed, the root is always labeled by $1$ and the binomial coefficient counts the number of ways to choose $|B|$ distinct values amongst $n-1$.
In~\cite{cleary2017some} (Lemma 2.3), the authors show that 
$$p(A) = \frac{n!}{\prod_{t \in A} (\text{number of nodes in the tree enrooted in }t)}$$
that is a Hook length formula on the Young tableaux associated to binary trees. 

\begin{lemma}\label{lm:avperm}
For a fixed sequence pattern of size $m$, assuming the text is drawn over the uniform distribution over permutations of size $n$, the average complexity of the \algoName{MetaAlgorithm}
is $\Theta(n)$.
\end{lemma}
\begin{proof}
In this case, Lemma~\nameref{lm:kappa} does not hold.
Given the Cartesian tree $A$ associated to a prefix of length $j-1$, such that
$\RMP(A)=1$, then the $j$-th node is a root with probability $1/j$ and a 
right-child with probability $(j-1)/j$. Therefore, there does not exists a constant
$\kappa$ that satisfies Lemma~\nameref{lm:kappa}.
Though, assuming the uniform distribution over permutations of size $n$ 
to generate the text, the average complexity of the $\eqTest$ algorithm
remains constant.
Given a cost $c$ that depends on the considered notion of pattern matching, the average number of comparisons performed by the $\eqTest$ function is bounded by the following formula:
$$c+\sum_{j=1}^{m} j \cdot \mathbb{P}(\LND_{x}[1\ldots j] = \LND_{p}[1\ldots j]) + \sum_{\ell=1}^{m} (m-\ell+1)\ \cdot \mathbb{P}(\LNG_{x}[\ell \ldots m] = \LNG_{p}[\ell \ldots m])$$
Since, in the case for permutations, for any fixed pattern $p$ and any factor $x$ of $t$, we have 
$\mathbb{P}(\LND_{x} = \LND_{p}) = \mathbb{P}(\LNG_{x} = \LNG_{p})$,
therefore the number of comparisons is less than:
$$c+2\sum_{j=1}^{m} j \cdot \mathbb{P}(\LND_{x}[1\ldots j] = \LND_{p}[1\ldots j])$$
Using the Hook length formula, we have that, for all $j\ge 4$, and all tree $A$ of size $j$, 
$\frac{p(A)}{j!} \le \frac{4}{j(j-2)(j-3)}$.  
Indeed, we have 
$$\frac{p(A)}{j!} = \frac{1}{\prod_{t \in A} (\text{number of nodes in the tree enrooted in }t)} $$
The tree $A$ counts for $j$. The two subtrees enrooted at level $1$ contributes by
$k \times ( j-1 -k)$, which is greater than $j-2$ for all $0\le k \le j-1$.
Finally, at level $2$, there exists at least one subtree which contains at least $\lceil \frac{j-3}{4} \rceil$
nodes. 
Therefore the average number of comparisons is bounded by
$$c+ 6 + 2\sum_{j=4}^{m} \frac{4}{(j-2)(j-3)}$$
Assuming the cost $c$ is different from $0$ if there exist a value $\ell$
such that $\LND_{x}[1\cdots \ell-1] = \LND_{p}[1\cdots \ell-1])$ and $\LNG_{x}[\ell+2\cdots m] = \LNG_{p}[\ell+2\cdots m]$.
Assuming also that $c$ is at most linear,
it can be proven that c has a constant contribution 
to the average cost of the function.
Since the sum itself tends to a constant, this concludes the proof.

\end{proof}

\section{Characterization of the parent-distance tables when a swap occurs \label{sec:charac}}
In this section, let $x$ be a sequence of length $m$,
 $i\in\{1,\ldots,m-1\}$ be an integer, and $y=\tau(x,i)$.

\subsection{Parent-distance tables}

We now describe the differences
 between the parent-distance tables of $x$, $\PDD_x$ and $\PDG_x$ and the parent-distance
 tables of $y$, $\PDD_y$ 
 and $\PDG_y$. 
\Cref{fig:zones} sums up the different notations and parts of the 
parent-distance tables we are going to characterize.


\begin{figure}[!ht]
\begin{center}
\includegraphics[scale=0.3]{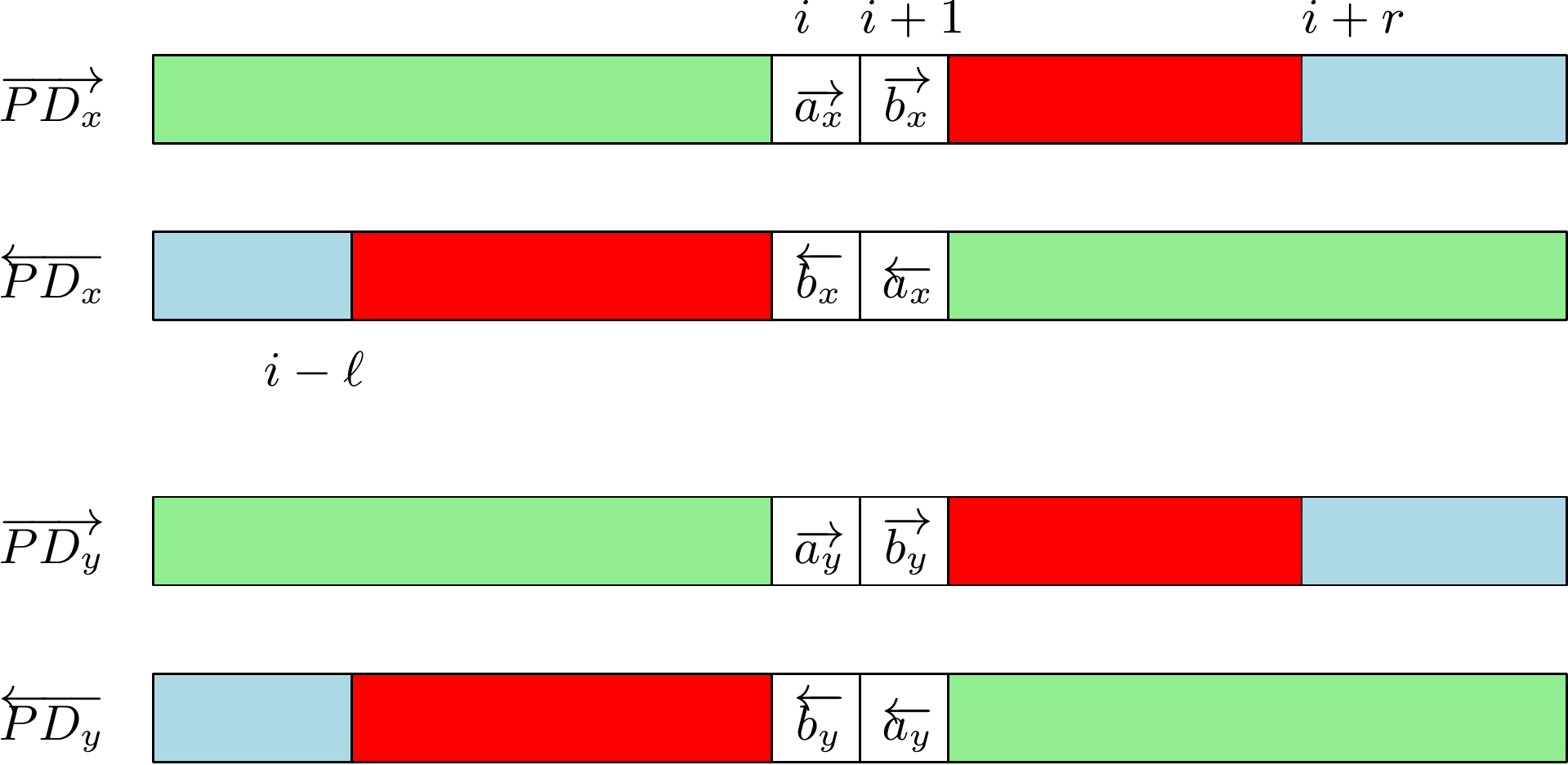}
\caption{This figure sums up the different Lemmas of this section. For instance, the green zones correspond to \nameCref{def:green}~\ref{def:green} and Lemma~\nameref{lm:green}. The values $\protect\overrightarrow{a_x}$, $\protect\overrightarrow{b_x}$, $\ldots$, are the $8$ values found in the parent-distance tables of $x$ and $y$ at position $i$ and $i+1$, that is $\protect\overrightarrow{PD_x}[i] = \protect\overrightarrow{a_x}$, $\protect\overrightarrow{PD_x}[i+1] = \protect\overrightarrow{b_x}$, $\ldots$ Values $i-\ell$ and $i+r$ respectively denote the last and first position of each blue zone.}\label{fig:zones}
\end{center}
\end{figure}

First, we describe how the parent-distances differ at positions $i$ and $i+1$.

\begin{lemma}\label{lm:lesser}
Suppose that $x[i] < x[i+1]$, then the following properties hold:
\begin{enumerate}
    \item\label{lm:lesser:item1} $\dG = 1$.
    \item\label{lm:lesser:item2} $\dD = \begin{cases}0 &\text{if }\aD = 0 \\ \aD + 1 & \text{otherwise.}\end{cases}$
    \item\label{lm:lesser:item3} $\cG = \begin{cases}0 &\text{if } \bG = 0\\ \bG - 1 & \text{otherwise.}\end{cases}$
    \item\label{lm:lesser:item4} $\cD \leq \begin{cases} i-1 & \text{if } \aD = 0\\ \aD & \text{otherwise.}\end{cases}$
\end{enumerate}
\end{lemma}

\begin{proof}
Suppose $x[i] < x[i+1]$, we have $\bD = 1$ by definition of the parent-distance (\Cref{def:pd}) and $\bG\neq1$ by definition of the reverse parent-distance
 (\Cref{rev-pd}). Then, if a swap occurs at position $i$, $y[i] > y[i+1]$ and we have:
\begin{enumerate}
    \item $\dG = 1$ by \Cref{rev-pd}.
    
    \item If $\aD = 0$,  
        $x[i]$ is the smallest element in $x[1\ldots i]$ by \Cref{def:pd}. Which implies $y[i+1]$ is the smallest element in $y[1\ldots i+1]$ and thus $\dD = 0$ by Definition~\ref{def:pd}.
        
        Otherwise, $x[i]$ (resp. $y[i+1]$) is not the smallest element in $x[1\ldots i]$ (resp. $y[1\ldots i+1]$). 
        $y[i+1]$ has been pushed away from its parent in $y[1\ldots i-1]$ by one position compared to $x[i]$ and its parent in $x[1\ldots i-1]$. 
        Thus, $\dD = \aD + 1$.
    
    \item If $\bG = 0$, 
        $x[i]$ is the smallest element in $x[i\ldots m]$ by \Cref{rev-pd}. Which implies $y[i+1]$ is the smallest element in $y[i+1\ldots m]$, and thus $\cG = 0$ by \Cref{rev-pd}.
        
        Otherwise, since $x[i]<x[i+1]$, we have $\bG > 1$ by \Cref{rev-pd} and $x[i]$ (resp. $y[i+1]$) is not the smallest element in $x[i\ldots m]$ (resp. $y[i+1\ldots m]$). 
        $y[i+1]$ has been pushed closer to its parent in $y[i+2\ldots m]$ by one position when compared to $x[i]$ and its parent in $x[i+2\ldots m]$. 
        Thus, $\cG = \bG - 1$.
    
    \item If $\aD > 0$, 
        that means there is an element smaller than $x[i]$ at position $i - \aD$ by Definition~\ref{def:pd}. After the swap, the parent-distance of $y[i]$ either refers to that same element at position $i - \aD$ or to a closer one that is smaller than $y[i]$ if such an element exists, and thus $\cD \leq \aD$.
        
        Otherwise, the only information we have is $\cD \leq i - 1$ by \Cref{def:pd}.
\end{enumerate}
\end{proof}

In the following lemma, we give all possibles values for~\cref{lm:lesser:item4} of \Cref{lm:lesser}.

\begin{lemma}\label{lm:precision}
Suppose that $x[i] < x[i+1]$, $\cD \in \{ i - pos \mid pos \in \rb(\leftsb(C_i(x)) \} \cup \{\aD \}$.
\end{lemma}

\begin{proof}
The node $i$ in $C(y)$ either has the same parent as in $C(x)$
or is inserted in the right branch of the left subtree of $C_i(x)$.
\end{proof}

\begin{lemma}\label{lm:greater}
Suppose that $x[i] > x[i+1]$, then the following properties hold:
\begin{enumerate}
    \item\label{lm:greater:item1} $\dD = 1$;
    \item\label{lm:greater:item2} $\dG = \begin{cases}0 &\text{if }\aG = 0 \\ \aG + 1 & \text{otherwise;}\end{cases}$
    \item\label{lm:greater:item3} $\cD = \begin{cases}0 &\text{if } \bD = 0\\ \bD - 1 & \text{otherwise;}\end{cases}$
    \item\label{lm:greater:item4} $\cG \leq \begin{cases} i-1 & \text{if } \aG = 0\\ \aG & \text{otherwise.}\end{cases}$
\end{enumerate}
\end{lemma}

In the following lemma, we give all possibles values for~\cref{lm:greater:item4} of \Cref{lm:greater}.

\begin{lemma}\label{lm:precision2}
Suppose that $x[i] > x[i+1]$, $\cG \in \{ pos - i \mid pos \in \lb(\rightsb(C_i(x)) \} \cup \{\aG \}$.
\end{lemma}

The proofs are similar to the ones of 
\Cref{lm:lesser,lm:precision}.

In the following, we will define the green and blue zones of the parent-distances
tables, which are equal, meaning that they are unaffected by the swap. Also, 
we define the red zones whose values differ by at most $1$.
We strongly invite the reader to use 
\Cref{fig:zones,fig:blue_red_appendix}
to get a better grasp of the definitions.

We first propose the following lemma to help with incoming proofs related
to the different zones. Informally, the idea is that for all positions/nodes
$j$ whose parent is neither $i$ nor $i+1$, 
the values of the parent-distance tables of $y$ at these positions $j$
should be be the same as 
the values of the parent-distance tables of $x$ at these positions $j$.

\begin{lemma}\label{lm:propzones}
  For all $j\in \{1, \ldots, m\} \setminus \{i,i+1\}$:
  \begin{itemize}
      \item 
      if $\PDD_x[j] \notin \{ j-i-1,j-i\}$, then $\PDD_y[j] = \PDD_x[j]$,
      \item
      if $\PDG_x[j] \notin \{ i-j, i+1-j\}$, then $\PDG_y[j] = \PDG_x[j]$.
  \end{itemize}        
\end{lemma}

\begin{proof}
  According to 
  \Cref{def:pd,rev-pd,def:swap}
  we have:
  \begin{itemize}
      \item for all $j<i$, $x[j] = y[j]$: 
      \begin{itemize}
          \item therefore $\PDD_y[j] = \PDD_x[j]$;
          \item if $\PDG_x[j] \notin \{ i-j, i+1-j\}$, then by definition
          for all $k \in \{j+1, \ldots, j+\PDG_x[j] -1\}$, both $x[j]$ and $y[j]$ are 
          smaller than $x[k]$ and greater than $x[j+\PDG_x[j]]$ (which is equal to $y[j+\PDG_x[j]]$).
          Therefore $\PDG_y[j] = \PDG_x[j]$.
      \end{itemize} 
      \item for all $j>i+1$, $y[j] = x[j]$:       
      the proof is similar to the previous item.
  \end{itemize}
\end{proof}
We now introduce the different zones and show how a swap at position $i$ affects them.
\begin{definition}[The green zones]\label{def:green}
The green zones of $\PDD_x$ and $\PDD_y$ are $\PDD_x[1\ldots i-1]$ and 
$\PDD_y[1\ldots i-1]$. 
The green zones of $\PDG_x$ and $\PDG_y$ are $\PDG_x[i+2\ldots m]$ and $\PDG_y[i+2\ldots m]$.
\end{definition}

\begin{lemma}[Green]\label{lm:green}
The green zones of $\PDD_x$ and $\PDD_y$ (resp. $\PDG_x$ and $\PDG_y$) are equal.
\end{lemma}

\begin{proof}
The proof directly follows from 
\Cref{lm:propzones,def:green}.
\end{proof}

\begin{definition}[The blue zones]\label{def:blue}
The blue zones of 
$\PDD_x$ and $\PDD_y$ are $\PDD_x[i+r \ldots m]$ and $\PDD_y[i+r \ldots m]$ where:
$$
    r = 
    \begin{cases}
    \bG  &\text{if } x[i]<x[i+1] \text{ and } \bG > 1 \\  
    \aG+1  &\text{if } x[i] > x[i+1] \text{ and } \aG >0\\
    m-i+1 &\text{otherwise.}
    \end{cases}
$$
The blue zones of $\PDG_x$ and $\PDG_y$ are $\PDG_x[1 \ldots i-\ell]$ and $\PDG_y[1 \ldots i-\ell]$ where:
$$
    \ell = 
    \begin{cases}
    \aD  &\text{if } x[i]<x[i+1] \text{ and } \aD>0 \\  
    \bD-1  &\text{if } x[i] > x[i+1] \text{ and } \bD>1\\
    i &\text{otherwise.}
    \end{cases}
$$
\end{definition}

Note that in the last cases, the blue zones are empty.

\begin{lemma}[Blue]\label{lm:blue}
The blue zones of $\PDD_x$ and $\PDD_y$ (resp. $\PDG_x$ and $\PDG_y$) are equal.
\end{lemma}

\begin{proof}
  Suppose $x[i] < x[i+1]$ (therefore $y[i] > y[i+1]$). 
  According to \Cref{def:blue}, the blue zones of $\PDD_x$ and $\PDD_y$ (resp. $\PDG_x$ and $\PDG_y$) are $\PDD_x[i+\bG\ldots m]$ and $\PDD_y[i+\cG +1\ldots m]$ (resp. $\PDG_x[1\ldots i - \aD]$  and $\PDG_y[1\ldots i - (\dD - 1)]$). 
  From \Cref{lm:lesser:item3}  of \Cref{lm:lesser}, 
  we have $\bG = \cG+1$, meaning that the blue zones of $\PDD_x$ and $\PDD_y$ coincide with each other (resp. \cref{lm:lesser:item2} of \Cref{lm:lesser} for  $\PDG_x$ and $\PDG_y$).
  
  Suppose $\bG > 1$, then there exists a position 
  $r=\bG$ such that for all $j \in \{i, \ldots, i+r-1\}$, we have $x[j] > x[i+r]$. 
  For each $k \in \{i+r, \ldots, m\}$, either $x[k] \le x[i+r]$, in which case $\PDD_x[k]$ and $\PDD_y[k]$ both point to the green zone and therefore did not change. Otherwise,
  $\PDD_x[k]$ and $\PDD_y[k]$ point to at least position $i+r$ and are therefore equal.
  
  If $x[i]<x[i+1]$ then by \Cref{rev-pd} it holds that 
   $x[i+r]<x[i]$ and $x[j]>x[i]$ for all $j\in\{i+1,\ldots,i+r-1\}$.
  Let $k\in\{i+r,\ldots,m\}$ be a position the blue zone of $\PDD_x$.
  If $x[k]>x[i+r]$ then $\PDD_x[k]\le k-i-r<k-i$ and then
   by \Cref{lm:propzones} we have $\PDD_y[k]=\PDD_x[k]$.
  If $x[k]<x[i+r]$ then $x[k]<x[j]$ for all $j\in\{i,\ldots,i+r\}$ and
   $\PDD_x[k]> k-i$ and then by \Cref{lm:propzones} we have $\PDD_y[k]=\PDD_x[k]$.
  
  The other cases can be proved in a similar way.
  
  The proof is similar for $x[i] > x[i+1]$.
\end{proof}

\begin{definition}[The red zones]\label{def:red}
If the blue zone of 
$\PDD_x$ is $\PDD_x[i+r \ldots m]$, then the right red zone is $\PDD_x[i+2 \ldots i+r-1]$.
Conversely, if the blue zone of $\PDG_x$ is $\PDG_x[1 \ldots i-\ell]$, then the left red zone is
$\PDG_x[i-\ell+1 \ldots i-1]$.
The same is true for $\PDD_y$ and $\PDG_y$.
\end{definition}

\begin{lemma}[Red]\label{lm:red}
For each position $j>i+1$ in the red zone of $\PDD_x$.
$$\PDD_y[j] = \begin{cases} \PDD_x[j]-1, \text{ if } \PDD_x[j]=j-i \\  \PDD_x[j]+1, \text{ if } \PDD_x[j]=j-i-1 \text{ and }  x[i] > x[j] > x[i+1] ,\\ \PDD_x[j], \text{ otherwise} \end{cases}$$
For each position $j<i$ in the red zone of $\PDG_x$.
$$\PDG_y[j] = \begin{cases} \PDG_x[j]-1, \text{ if } \PDG_x[j]=i+1-j \\  \PDG_x[j]+1, \text{ if }  \PDG_x[j]=i-j \text{ and } x[i] < x[j] < x[i+1] \\ \PDG_x[j], \text{ otherwise} \end{cases}$$
\end{lemma}

\begin{proof}
We only prove the Lemma for positions $j>i+1$ in the red zone of $\PDD_x$, since the logic is exactly the same
for $\PDG_x$.
According to Lemma~\ref{lm:propzones}, if $\PDG_x[j] \neq \PDG_y[j]$ then $\PDG_x[j] \in \{j-i, j-i-1\}$.
\begin{itemize}
    \item if $\PDG_x[j] = j-i$, then the parent of $j$ is $i$ and necessarily, $x[i]<x[i+1]$. When the swap is applied, the parent of $j$ is moved one position closer, therefore  $\PDG_y[j] = \PDG_x[j]-1$.
    \item if $\PDG_x[j] = j-i-1$, then the parent of $j$ is $i+1$. When the swap is applied, there are two possibilities.  Either $x[i] < x[j]$, therefore $y[i+1]<x[j]$ and $\PDG_y[j]=\PDG_x[j]$, or 
    $x[i] > x[j]$, then the parent of $j$ is moved one position further, and $\PDG_y[j] = \PDG_x[j]+1$
\end{itemize}
\end{proof}

\begin{figure}[!ht]
  \begin{minipage}[t]{0.49\textwidth}
      \vspace{0pt}
      \includegraphics[scale=0.35]{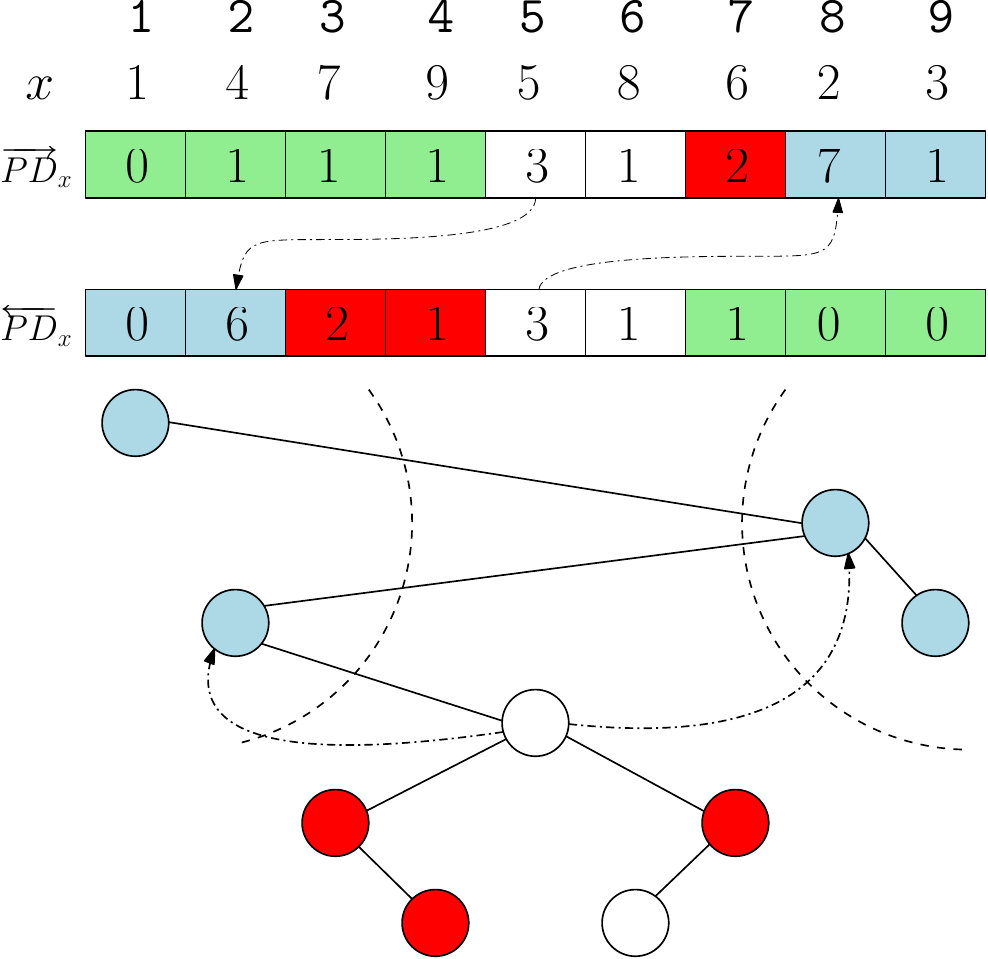}
  \end{minipage}
  \hfill
  \begin{minipage}[t]{0.49\textwidth}
      \vspace{0pt}
      \includegraphics[scale=0.35]{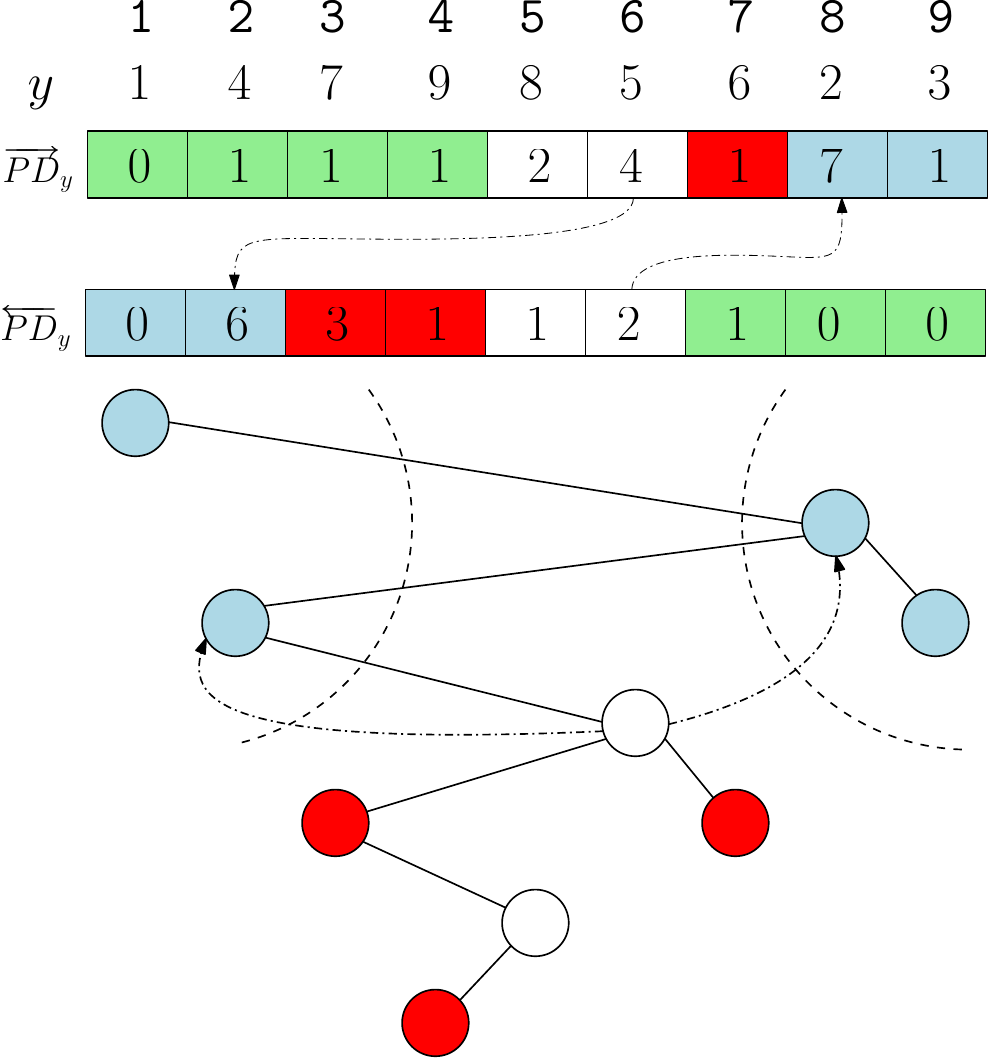}
  \end{minipage}
  \caption{\label{fig:blue_red_appendix}%
  In this figure, swaps are applied at position \texttt{5} on both $x$ and $y$.
  As can be seen on the left part of the Figure, $x[\texttt{5}] < x[\texttt{6}]$, $\ell=\PDD_x[\texttt{5}]$ and $r=\PDG_x[\texttt{5}]$ gives us the position $\texttt{5}+r$ and $\texttt{5}-\ell$ of the first values that are smaller than $x[\texttt{5}]$ and, by extension, smaller than any value between $\texttt{5}-\ell$ and $\texttt{5}+\ell$. Therefore any position smaller than $\texttt{5}-\ell$ in $\PDG_x$ is unaffected by the swap. The same
  goes for any position greater than $\texttt{5}+r$ in $\PDD_x$.
  On the right part of the figure, we have $y[\texttt{5}] > y[\texttt{6}]$, $\ell=\PDD_y[\texttt{6}]-1$ and $r=\PDG_y[\texttt{6}]+1$.}
\end{figure}
We now show that swaps at different positions produce different Cartesian trees.

\begin{lemma}\label{lm:uniondisjointe}
Let $j \in \{1,\ldots, m-1\}$ with $i\neq j$. 
Then $\tau(x,i) \nCTM \tau(x,j)$.
\end{lemma}

\begin{proof}
Suppose without loss of generality that $j > i $.
If $j>i+1$, then according to Lemma~\nameref{lm:green}, we have:
$$\forall k< j, \PDD_x[k] = \PDD_{\tau(x,j)}[k] = \PDD_{\tau(x,i)}[k] $$
And according to \Cref{lm:lesser}, we have that $\PDD_x[i+1] \neq \PDD_{\tau(x,i)}[i+1]$,
which leads to a contradiction.

Then suppose that $j=i+1$, then it is sufficient to consider what happens on a sequence of length $3$:
having local differences on the parent-distance tables implies having different parent-distances and therefore not $CT$ matching. One can easily check that the lemma is true for each sequence of length $3$.
\end{proof}

\subsection{Computing the position of the swap with the green zones}

We  first show how to compute the position of the possible swap thanks to the green zones.
This is the idea behind the function \algoName{computeCandidates} in line~\ref{algo:testequivPD:linei} of algorithm~\ref{algo:testequivPD}.
If the parent-distances are equal, then the two sequences trivially $CT$ match.
Otherwise, the idea is to rely on a ``pincer movement'' using the green zones. 
According to Lemma~\nameref{lm:green}, the green zones of
$\PDD_x$ and $\PDD_y$ (resp. $\PDG_x$ and $\PDG_y$) are equal. 
From \Cref{lm:lesser,lm:greater}, we also deduce that $\bD \neq \dD$ and $\bG \neq \dG$. 
Unfortunately, we might run into four different cases, depending on whether $\aD$ equals $\cD$ and $\aG$ equals $\cG$. 
\Cref{fig:indistinguability} presents all four possible cases.

\begin{figure}[ht!]
    \centering
    \includegraphics[scale=0.55]{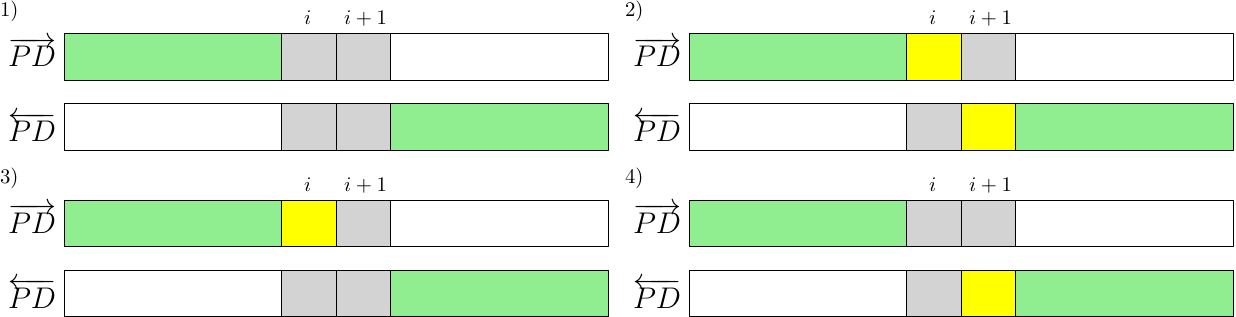}
    \caption{The parent-distance tables of sequences $x$ and $y$ are merged into one. If a zone is colored either in green or yellow, then the tables match, in grey if they do not and white when it is unknown.}
    \label{fig:indistinguability}
\end{figure}

In case 1 of ~\Cref{fig:indistinguability}, there is a gap of length 2 between the green/yellow zones, and we can immediately deduce that the positions of this gap are the only eligible positions for the swapped elements. In case 2, the gap is reduced to 0, but we can once again pinpoint the position $i$ with total accuracy. Lastly, in cases 3 and 4, there is a gap of 1,
and as we currently have no knowledge of the position of the swap, both cases end up indistinguishable,
meaning that we must test the positions on the left of the gap and of the gap itself. From~\Cref{lm:uniondisjointe}, we have that only one such position 
might be a swap, as two swaps on the same sequence would produce different Cartesian trees and thus different parent-distance tables.
Any gap larger than 2 immediately disqualifies the sequences from $CT_\tau$ matching.

Algorithm~\ref{algo:candidates} computes the location of the possible swap, using the green zones as described above. Its parameters $j$ and $k$ correspond to the first indexes at which the parent-distance (resp. reverse parent-distance) tables differ, that is the first encountered grey areas in~\Cref{fig:indistinguability}.

\begin{algorithm}\label{algo:candidates}
  \caption{\algoName{computeCandidates}$(i, j)$}
  \SetAlgoLined
  \SetKwInOut{KwIn}{Input}
  \SetKwInOut{KwOut}{Output}
  \KwIn{Two positions $j$ and $k$}
  \KwOut{The set of positions where a swap may have happened}
  $d \leftarrow k-j$\;
  \If(\tcp*[f]{see case 1 in \Cref{fig:indistinguability}}){$d = 1$}{ 
    return $\{j\}$\; 
  }
  \If(\tcp*[f]{see case 2 in ~\Cref{fig:indistinguability}}){$d = -1$}{
    return $\{k\}$\;
  }
  \If(\tcp*[f]{see cases 3, 4 in ~\Cref{fig:indistinguability}}){$d = 0$}{
    return $\{j-1, j\}$\;
  }
  return $\emptyset$\;
\end{algorithm}

\subsection{The {\em equivalenceTest} function}

Algorithm~\ref{algo:testequivPD}, below, is based on
\Crefrange{lm:lesser}{lm:uniondisjointe}. It takes the parent-distance and reverse parent-distance tables of the pattern and the current window on the text as inputs and returns $True$ if they $CT_\tau$ match and $False$ otherwise.

\begin{algorithm}\label{algo:testequivPD}
  \caption{\algoName{equivalenceTestPD}$((\PDD_p, \PDG_p), (\PDD_x, \PDG_x))$}
  \SetAlgoLined
  \SetKwInOut{KwIn}{Input}
  \SetKwInOut{KwOut}{Output}
  \KwIn{The parent-distance tables of $p$ and $x$}
  \KwOut{$True$ if $p \tauCTM x$, $False$ otherwise}
  $j \leftarrow 2$\;
  \While{$j \le m$ and $\PDD_p[j] = \PDD_x[j]$ }{
  $j \leftarrow j+1$\;
  }
  \If(\tcp*[f]{Exact match}){$j=m+1$}{
        \Return $True$\; 
  }
  $k = m - 1$\;
  \While{$k \ge j$ and $\PDG_p[k] = \PDG_x[k]$ }{
  $k \leftarrow k-1$\;
  }
  $candidatePositions \leftarrow \algoName{computeCandidates}(j, k)$\; \label{algo:testequivPD:linei}
  \ForEach{$pos \in candidatePositions$}{
    \If{Lemmas~\ref{lm:lesser},~\ref{lm:greater},~\nameref{lm:blue} and~\nameref{lm:red} hold for $p, x$ and $pos$}{
    \Return{$True$}\;
    }
  }
  \Return{$False$}\;
\end{algorithm}

\begin{thm}
Given two sequences $p$ and $x$ of length $m$, Algorithm~\ref{algo:testequivPD}
 has a $\Theta(m)$ worst-case time complexity and a $\Theta(1)$ best-case complexity
 and a $\Theta(1)$ space complexity.
\end{thm}

\subsection{Updating the parent-distance and reverse parent-distance tables}

When searching for a pattern $p$ of length $m$ in a text $t$
 of length $n$, the parent-distance representations of $p$ are computed once in
 a preprocessing phase.
The searching phase looks at the text $t$ through a window of size $m$.
The parent-distance representations of $t[1\ldots m]$
 are first computed, then for $m+1\le j \le n$ the parent-distance representations
 of $t[j-m+1\ldots j]$ are computed by updating the parent-distance representations
 of $t[j-m\ldots j-1]$.
For that, function \algoName{updatePD} in Algorithm~\ref{algo:updatePD} uses:
\begin{itemize}
    \item
    a deque $D$
    for storing the right branch of the Cartesian tree of the current window on $t$. Each element of $D$ consists of  a pair $(val,pos)$ where $val$ is a symbol of $t$ and $pos$ is the associated position;
    \item
    two circular arrays $\PDD$ and $\PDG$ of size $m$ for storing the Parent-distance representations
    of the current window on $t$;
    \item
    a circular array $\notREF$ of size $m$ for storing for
     each position $k$ in the current window, the positions of which $k$ is the referent.
     More formally if $\REFD{{t[j-m\ldots j-1]}}(h) = k$ then $h \in \notREF[k]$.
     Array $\notREF$ is equivalent to table $D$ in~\cite{demaine09}.
\end{itemize}

If the first element of $D$ is equal to $j-m$ then it is removed from $D$. For each $pos \in \notREF[j - m]$, we know they are ``losing their parent'' as we slide the window and we set their parent-distance to 0. We set $\PDD[j - m + 1]$ to 0 according to~\cref{def:pd}, since it is the new ``first'' element after the update. 
We compute the value of $\PDD[j]$ thanks to $D$ as in~\cite{PALP19} and also update $\notREF$ accordingly. 
By~\cref{rev-pd}, we have $\PDG[j] = 0$. 
Lastly, with the addition of node $j$, some reverse-parent distances who were previously equal to 0 may have ``gained a parent'' in $j$ and need to be updated accordingly. 
All said reverse-parent distances are the positions on $\rb(\leftsb(C_j(t[j-m+1\ldots j]))$, that is all the positions that were removed from $D$ and added to the checklist while computing $\PDD[j]$. 

\begin{algorithm}\label{algo:updatePD}
  \caption{\algoName{updatePD}$(\PDD, \PDG, t, j, m, \notREF, D)$}
  \SetAlgoLined
  \SetKwInOut{KwIn}{Input}
  \SetKwInOut{KwOut}{Output}
  \KwIn{$t$: sequence, %
        $j$: position, %
        $m$: window size, %
        $D$: right branch of $C(t[j-m\ldots j-1]$, %
        $\PDD, \PDG$:
        parent-distance representations of $t[j-m\ldots j-1]$,
        $\notREF$: array with information pertaining to the referents in $t[j-m\ldots j-1]$}
  \KwOut{$(D, \PDD, \PDG, \notREF)$ such that %
         $D$: right branch of $C(t[j-m+1\ldots j]$, %
         $\PDD, \PDG$: parent-distance representations of $t[j-m+1\ldots j]$, %
         $\notREF$: array with information pertaining to the referents in $t[j-m+1\ldots j]$}
  $checklist \leftarrow \emptyset$\;
  
  \ForEach{$pos \in \notREF[((j-m) \bmod m)+1]$}{
    $\PDD[(pos \bmod m)+1] \leftarrow 0$\;
  }
  $(val, pos) \leftarrow \dBack(D)$\;
  \If{$pos = j-m$}{
    $\dPopBack(D)$\;
  }
  $\PDD[((j + 1) \bmod m)+1] \leftarrow 0$\;
  \While{$not\ \dIsEmpty(D)$}{
    $(val, pos) \leftarrow \dFront(D)$\;
    \If{$val \leq t[j]$}{
        $break$\;
    }
    $checklist \leftarrow checklist \cup \{pos\}$\;
    $\dPopFront(D)$\;
  }
  \If{$not\ \dIsEmpty(D)$}{
    $\PDD[(j \bmod m) + 1] \leftarrow 0$\;
  }
  \Else{
    $\PDD[(j \bmod m) + 1] = j - pos$\;
    $\notREF[(pos \bmod m)+1] \leftarrow \notREF[(pos \bmod m)+1] \cup \{j\}$\;
  }
  $\dPushFront(D, (\PDD[(j \bmod m) + 1], j))$\;
  $\PDG[(j \bmod m) + 1] \leftarrow 0$\;
  \If{$\PDD[(j \bmod m) + 1] \neq 1$}{
    \ForEach{$pos \in checklist$}{
        $\PDG[(pos \bmod m) +1] \leftarrow j - pos$\;
    }
  }
\end{algorithm}

\section{Swap graph of Cartesian trees \label{sec:graph}}
\subsection{Swap graph}

In this section, we define a graph of Cartesian trees, where two trees are connected
by an edge if one can be obtained from the other using a swap operation.

\begin{definition}[neighbours and neighbourhood, $\ngi$ and $\ng$]\label{def:ng}
Let $T \in \mathcal{C}_m$ be a Cartesian tree with $m$ nodes.
We define $\ngi(T,i)$ the set of Cartesian trees $C(y)$ obtained by 
identifying a sequence $x$ such that $T=C(x)$ and doing a swap on $x$ at position $1\leq i\leq m-1$, that is:
$$
    \ngi(T,i) = \{ C(y) \in \mathcal{C}_m \mid \exists\ x\text{ such that  }T=C(x) \text{ and } y=\tau(x,i)  \}
$$
\end{definition}
Also, we have
$$
    \ng(T) = \bigcup_{i=1}^{m-1} \ngi(T,i)
$$
where all unions are disjoint according to Lemma~\ref{lm:uniondisjointe}. Informally, we will say that $\ngi(T, i)$ are the neighbours of $T$ with a swap at position $i$ and we will call $\ng(T)$ the neighbourhood of $T$.

\begin{definition}[Swap graph $\SGraph_m$]
Let $m$ be an integer. 
The \emph{Swap graph} of Cartesian trees of size $m$, denoted by 
$\mathcal{G}_m=(\mathcal{C}_m, E_m)$,
where $\mathcal{C}_m$ is its set of vertices, and $E_m$ the set of edges such that
 $\{C(x),C(y)\} \in E_m$ if $C(y) \in \ng(C(x))$.
\end{definition}
\Cref{fig:graph} shows the Swap graphs $\SGraph_m$ with $m$ smaller than $4$. 

\begin{figure}[ht!]
\begin{center}
\begin{minipage}{0.49\textwidth}
\begin{center}
    \includegraphics{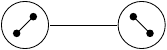}\\
    \vspace{5em}

\includegraphics{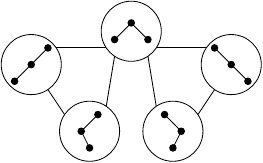}
\end{center}

\end{minipage}
\begin{minipage}{0.5\textwidth}
\includegraphics[scale=0.9]{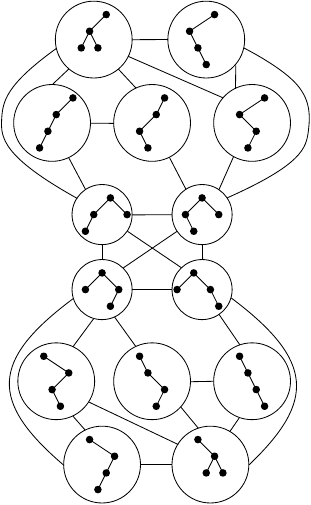}    
\end{minipage}
\end{center}
\caption{Swap graph of Cartesian trees of size $2, 3$ and $4$. \label{fig:graph}}
\end{figure}

In the following, we study the set of neighbours a vertex can have in the Swap graph.
Let $T \in \mathcal{C}_m$ be a Cartesian tree of size $m$ and $\ng(T)$ be its neighbourhood in the Swap graph.

\begin{lemma}\label{lm:decomp}
    Let $T \in \mathcal{C}_{m}$ be a Cartesian tree of size $m$
    with 
    $\leftsb(T)$ of size $k-1$ and 
    $\rightsb(T)$ of size $m-k$.
    We have
    $$
        |\ng(T)| = |\ng(\leftsb(T))| + |\ng(\rightsb(T))| + |\ngi(T,k-1)| + |\ngi(T,k)|
    $$
\end{lemma}
\begin{proof}
    The result follows from Lemma~\ref{lm:uniondisjointe} and the definition of $\ng(T)$ (\Cref{def:ng}).
    Indeed, we have 
        $|\ng(\leftsb(T))| =   |\bigcup_{i=1}^{k-2}  \ngi(T,i)|$ and 
        %
        $|\ng(\rightsb(T))| = |\bigcup_{i=k+1}^{m-1} \ngi(T,i)|$.
\end{proof}

\begin{lemma}\label{lm:decompT}
    Let $T \in \mathcal{C}_{m}$ be a Cartesian tree of size $m$ with a root $i$, 
    $$
        |\ngi(T, i-1)| = \LMP(\rightsb(T))+1 \text{ and } |\ngi(T,i)| = \RMP(\leftsb(T))+1 
    $$
\end{lemma}
\begin{proof}
    Let $B=\rightsb(T)$.
    We only prove that $|\ngi(T,i-1)| = \LMP(B)+1$ since the rest of the proof uses the same arguments.
    Let $x$ be a sequence such that $C(x) = T$.
    As stated in the definition section (see \Cref{fig:transposition}), the swap $\tau(x,i-1)$ moves the rightmost node of $\leftsb(T)$ into a leftmost position in $B$. 
    Let $j_1, \ldots, j_{\LMP(B)}$ be the positions
    in the sequence $x$ that corresponds to the nodes of the left branch of $B$. For each $\ell < \LMP(B)$, there always
    exists a sequence $y = \tau(x,i)$ such that $y[i] < y[j_1] < \cdots < y[j_\ell] < y[i-1] < y[j_{\ell+1}] < \cdots < y[j_{\LMP(B)}]$. Therefore, there exist exactly $\LMP(B)+1$ possible output trees when applying such a swap. 
\end{proof}

This link between the number of possible outputs for a swap at a given position and the length of a rightmost (resp. leftmost) path in a subtree is given by the Skipped-number (resp. reverse Skipped-number) representation.

\begin{lemma}\label{lm:neighbor}
For every Cartesian tree $T \in \mathcal{C}_{m} $ of size $m \geq 2$, we have 
    $$
        m - 1 \le |\ng(T)| \le 3(m-2)+1
    $$
\end{lemma}

\begin{proof}
Let us first consider the following claim that is easily verified by \Cref{def:SN}:
$$\forall m \geq 2, \sum_{j=1}^{m} \SN[j] \leq m - 1.$$

And its converse, where $T_j$ is the subtree of $T$ enrooted in $j$:
$$\forall \ell < m, \sum_{j=\ell}^{m} \LMP(\rightsb(T_j))  \leq m - \ell - 1$$

From these two inequalities and \Cref{lm:decompT}, we also immediately have:
$$
    \forall j \leq m - 1, |\ngi(T, j)| = 
    \begin{cases}
        \LMP(\rightsb(T_{j+1})) + 1,\ if\ x[j] > x[j+1]\\
        \SN[j] + 1,\ otherwise
    \end{cases}
$$

Now, let us recall that, according to \Cref{lm:uniondisjointe}, we have
$$
    |\ng(T)| = \sum_{j=1}^{m-1} |\ngi(T,j)|
$$

This gives us the following upper bound on the size of the neighbourhood:

\begin{align*}
|\ng(T)| & \leq \sum_{j=1}^{m-1} \left(\LMP(\rightsb(T_{j+1})) + \SN[j] + 1\right)\\
         & \leq \sum_{j=1}^{m-1} \LMP(\rightsb(T_{j+1})) + \sum_{j=1}^{m-1} \SN[j] + \sum_{j=1}^{m-1} 1\\
         & \leq \sum_{j=2}^{m} \LMP(\rightsb(T_{j+1})) + \sum_{j=1}^{m-1} \SN[j] + m - 1\\
         & \leq (m - 2) + (m - 2) + m - 1\\
         & \leq 3(m-2)+1
\end{align*}

\end{proof}
We use the previous lemma to obtain a lower bound on the diameter of the Swap graph.

\begin{lemma}
    The diameter of the Swap graph $\SGraph_m$ is $\Omega(\frac{m}{\ln{m}})$.
\end{lemma}
\begin{proof}
    The number of vertices in the graph is equal to the number of binary trees enumerated by the Catalan numbers, that is $\frac{\binom{2m}{m}}{m+1}$. Since the maximal degree of a vertex is less than $3m$ according to \Cref{lm:neighbor}, the diameter is lower bounded by the value $k$ such
    that:
    $$(3m)^k = \frac{\binom{2m}{m}}{m+1}$$
    $$\implies k = \frac{\ln\left(\frac{\binom{2m}{m}}{m + 1} \right)}{\ln(3) + \ln(m)}$$
    $$\implies k = \frac{2m\ln(2m) - 2m\ln(m) - \ln(m + 1)}{\ln(3) + \ln(m)}$$
    By decomposing $2m\ln(2m)$ into $2m\ln{2} + 2m\ln{m}$ we obtain
    $$\implies k = \frac{2m\ln(2)}{\ln(3) + \ln(m)} - \frac{\ln(m + 1)}{\ln(3) + \ln(m)}$$
which corresponds to the announced result.
\end{proof}

\subsection{An Aho-Corasick based algorithm}

The idea of the following method is to take advantage of the upper bound on
the size of the neighbourhood of a given Cartesian tree in the Swap graph.
Given a sequence $p$, we compute the set of its neighbours $\ng(C(p))$,
then we compute the set of all parent-distance tables and build the 
automaton that recognizes this set of tables using the Aho-Corasick method
 for multiple Cartesian tree matching~\cite{PALP19}.
Then, it is sufficient to read the parent-distance table of the text into the automaton
and check whenever we reach a final state.



In order to compute the neighbourhood of a given tree $T$ of size $m$, we compute the parent-distance tables of every set of neighbours of said tree if a swap occurs at position $i$, for all $1 \leq i \leq m -1$. We need only distinguish two cases for every position $i$, whether $x[i] < x[i+1]$ or not. 

\begin{algorithm}\label{algo:ngPD}
  \caption{\algoName{buildAhoCorasickAutomaton}}
  \SetAlgoLined
  \SetKwInOut{KwIn}{Input}
  \SetKwInOut{KwOut}{Output}
  \KwIn{The parent-distance table $\PDD_x$ of a sequence $x$ of length $m$}
  \KwOut{The Aho-Corasick Automaton that recognizes $\{\PDD_x\} \cup\{\PDD_y \mid C(y) \in \ng(C(x))  \} $}
  $\mathcal{A} \leftarrow$ Compute a trie that recognizes $\PDD_x$\;
  \For{$i \in \{1,\ldots, m\}$}{
    $\ng \leftarrow $ Compute $\ng(C(x),i)$ according to 
    \Cref{lm:lesser,lm:precision,lm:greater,lm:precision2} and \nameref{lm:red}\;
    \ForEach{$\PDD_y \in \ng$}{
       Add $\PDD_y$ in $\mathcal{A}$\;
    }
  }
  Compute the failure function in $\mathcal{A}$\;
  Return $\mathcal{A}$ \;
\end{algorithm}
Lines $1$ and $5$ in Algorithm~\ref{algo:ngPD} is the classical method to add a word in the language recognized by a trie.
Line $6$ can be computed using~\cite{PALP19}.

The Aho-Corasick automaton contains at most $\mathcal{O}(m^2)$ states.
The following theorem can be obtained from Section $4.2$ in~\cite{PALP19}.
\begin{thm}
Given two sequences $p$ and $t$ of length $m$ and $n$, the Aho-Corasick based
 algorithm (Algorithm~\ref{algo:ngPD}) has an $\mathcal{O}((m^2 + n)\log(m) )$ worst-case time complexity
 and an $\mathcal{O}(m^2)$ space complexity.
\end{thm}

\section{Skipped-number representation when one swap occurs \label{sec:skipped}}
\newcommand{\SNl}{\overleftarrow{\mathit{SN}}}
\newcommand{\GS}{\mathit{GS}}


Again, in this section, let $x$ be a sequence of length $m$,
 $i\in\{1,\ldots,m-1\}$ be an integer, and $y=\tau(x,i)$.
It is divided in two parts.
The first one 
characterizes the differences between the \SNRepres\
of $x$ and the \SNRepres\ of $y$.
 The second part explains
how to update the \SNRepres\ of a text factor.

\subsection{Skipped-number tables}

In this subsection, we show
that the \SNRepres\
of $x$ and the \SNRepres\ of $y$ can differ in at most $3$ positions (\Cref{lm:swaptomismatch}).

We pinpoint the possible locations of those changes (\Cref{lm:pos,lm:swaptomismatch})
and finally, we show that we can precisely determine how the values in those positions change
by looking at a constant number of information (\Cref{lm:changesLess,lm:changesMore}).

We start by characterizing the positions where the \SNRepres\ of $y$ is equal to the \SNRepres\ of $x$.
Recall that it can be assumed that the sequences are totally ordered. In the case
of a partial order, one can linearize the
sequence in order to obtain a total ordering.
However, by Definition~\ref{def:swap}, for any sequence $x[1\ldots m]$ and $i \in \{1,\ldots, m-1\}$ such that $x[i]=x[i+1]$ it is not considered a swap.

\begin{lemma}\label{lm:pos}
  $\SN_y[j] = \SN_x[j]$ for all position $j \le m$ such that 
  $$
    j \notin \{i, i+1, \REFD{x}(i), \REFD{x}(i+1), \REFD{y}(i), \REFD{y}(i+1)\}.
  $$ 
\end{lemma}

\begin{proof}
Recall that a swap is an operation that either moves node $i$ from the rightmost path of the 
left subtree of node $i+1$ to a leftmost path of its right subtree or moves node $i+1$ from 
the leftmost path of the right subtree to a rightmost path of the left subtree of node $i$.

According to ~\cref{def:SN}, the \SNRepres\ only changes on the positions $j$ 
where $\RMP(\leftsb(C_j(x)))$ is modified, that is the number of nodes on the rightmost path of the left subtree of $j$ is modified. Positions $i$ and $i+1$ might be modified because their left subtree might be.
A position $j \notin \{i,i+1\}$ can only be affected if either $i$ or $i+1$ is in 
$\rb(\leftsb(C_j(x)))\cup \rb(\leftsb(C_j(y)))$. By~\cref{def:ref}, this position is in 
$\{\REFD{x}(i), \REFD{x}(i+1), \REFD{y}(i), \REFD{y}(i+1)\}$.
\end{proof}

Next, we show that 
the \SNRepres\ of $y$ cannot differ from 
the \SNRepres\ of $x$
for two distinct positions in 
$\{\REFD{x}(i), \REFD{x}(i+1), \REFD{y}(i), \REFD{y}(i+1)\}$.
This implies that there are at most $3$ mismatches between $\SN_x$ and $\SN_y$.

\begin{lemma}
\label{lm:swaptomismatch}
There exist at most $3$ positions $j$ such that $\SN_x[j] \neq \SN_y[j] $, where $j\in\{i, i+1, \REFD{x}(i), \REFD{x}(i+1), \REFD{y}(i), \REFD{y}(i+1)\}$.
\end{lemma}
\begin{proof}
    The aim of this proof is either to show that two positions are equal in the set or to show that there cannot be
    differences in the \SNRepres of $x$ and in the \SNRepres of $y$ at these positions.
    We distinguish two cases:
    \begin{enumerate}
        \item $\REFD{x}(i) = \REFD{x}(i+1)$ \label{lem:proof:refi_eq_refi1} (see in \Cref{fig:1}: $\tau(x,5)$): 
            let $j=\REFD{x}(i)$,
            if $j\neq -1$ this implies that $x[i]$ and $x[i+1]$ are both on the rightmost path of $C(x[1\ldots j-1])$. 
            If $j=-1$ then this implies that $x[i]$ and $x[i+1]$ are both on the rightmost path of $C(x[1\ldots m])$. 
            Hence, in both scenarios $x[i]<x[i+1]$.
            It also implies that $\REFD{x}(i+1) = \REFD{y}(i+1)$ and $\REFD{y}(i) = i+1$. 
            Therefore, the only positions for which the $\SN_y$ table can
            be different from the $\SN_x$ table are $i$, $i+1$, and $\REFD{x}(i)$.
        
        \item $\REFD{x}(i) \neq \REFD{x}(i+1)$: this case implies the following two scenarios.
        \begin{enumerate}
            \item
            $x[i] < x[i+1]$ (see in~\Cref{fig:1}: $\tau(x,4)$) \label{lem:proof:refi_neq_refi1_i_lti_1}: 
            since $y[i+1]<y[i]$, then we have $\REFD{y}(i) = i+1$.
            By~\Cref{def:ref} we have $\REFD{x}(i)=\REFD{y}(i+1)$.
            Since $\REFD{x}(i) \neq \REFD{x}(i+1)$, we have $x[\REFD{x}(i)]<x[i]<x[\REFD{x}(i+1)]<x[i+1]$ then
            $x[\REFD{x}(i)]<y[i+1]<x[\REFD{x}(i+1)]<y[i]$ then
            $x[i+1]\not\in\rb(C(x[1\ldots \REFD{x}(i)-1]))$ and
            $y[i]\not\in\rb(C(y[1\ldots \REFD{x}(i)-1]))$ and thus
            $\RMP(C(x[1\ldots \REFD{x}(i)-1])) = 
            \RMP(C(y[1\ldots \REFD{x}(i)-1]))$ 
            and thus $\SN_x[\REFD{x}(i)] = \SN_y[\REFD{x}(i)]$.
            Therefore, the only positions for which the $\SN_y$ table can
            be different from the $\SN_x$ table are $i$, $i+1$, and $\REFD{x}(i+1)$.%
            \item $x[i] > x[i+1]$ (see in \Cref{fig:1}: $\tau(x,2)$ and $\tau(x,3)$)\label{lem:proof:refi_neq_refi1_i_gti_1}:
            By~\cref{def:ref}, we know that $\REFD{x}(i)=i+1$.
            Let us recall that node $i+1$ is the last node added to the rightmost path of $C_{i+1}(x)$ and $C_{i+1}(y)$.
            For position $\REFD{x}(i+1)$, we have two options, either node $i+1$ will be later skipped by another node or not:
            \begin{enumerate}
                \item  \label{lem:proof:refi_neq_refi1_i_gti_1_i}               
                If node $i+1$ is skipped, then $\REFD{x}(i+1)\neq -1$ (see in \Cref{fig:1}: $\tau(x,2)$), we have $\REFD{x}(i+1) = \REFD{y}(i)$. 
                    Suppose we have $\REFD{y}(i+1) = \REFD{y}(i) = \REFD{x}(i+1)$.
                    In that case, the only positions for which the $\SN_y$ table can differ from the $\SN_x$ table are $i$, $i+1$, and $\REFD{y}(i+1)$.
                    Otherwise, this implies
                    $y[\REFD{y}(i) = \REFD{x}(i+1)]<y[i]<y[\REFD{y}(i+1)]<y[i+1]$ then
                    $x[\REFD{x}(i+1)]<x[i+1]<x[\REFD{y}(i+1)]<x[i]$ then
                    $x[i]\not\in\rb(C(x[1\ldots \REFD{y}(i)-1]))$ and
                    $y[i+1]\not\in\rb(C(y[1\ldots \REFD{y}(i)-1]))$ and thus
                    $\RMP(C(x[1\ldots \REFD{y}(i)-1])) = 
                    \RMP(C(y[1\ldots \REFD{y}(i)-1]))$ 
                    and thus $\SN_x[\REFD{y}(i)] = \SN_y[\REFD{y}(i)]$.
                    And once again, the only positions for which the $\SN_y$ table can differ from the $\SN_x$ table are $i$, $i+1$, and $\REFD{y}(i+1)$.%
                \item \label{lem:proof:refi_neq_refi1_i_gti_1_ii}               
                    If node $i+1$ isn't skipped, then $\REFD{x}(i+1)=-1$ (see in \Cref{fig:1}: $\tau(x,3)$) and it implies that $\REFD{y}(i)=-1$. 
                    Hence, the only positions for which the $\SN_y$ table can differ from the $\SN_x$ table are $i$, $i+1$, and $\REFD{y}(i+1)$.%
            \end{enumerate} 
        \end{enumerate}
    \end{enumerate}
\end{proof}

\begin{figure}[ht!]
            %
            \begin{minipage}{0.49\textwidth}
                \ref{lem:proof:refi_eq_refi1}\\
                \includegraphics[scale=0.7]{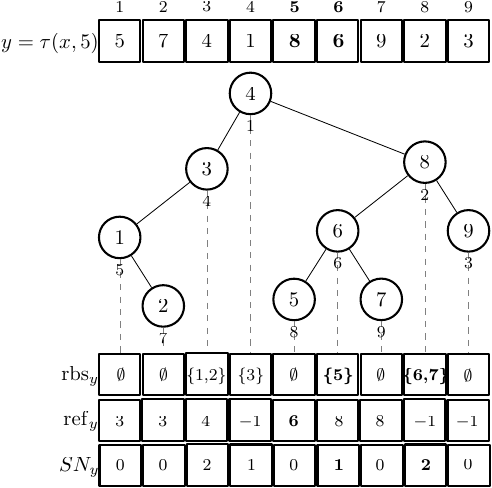} 
            \end{minipage}
            \begin{minipage}{0.49\textwidth}
                \ref{lem:proof:refi_neq_refi1_i_lti_1}\\
                \includegraphics[scale=0.7]{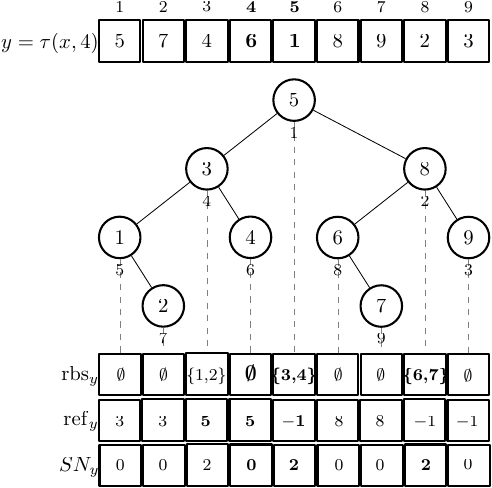} 
            \end{minipage}
            \begin{minipage}{0.49\textwidth}
               \ref{lem:proof:refi_neq_refi1_i_gti_1_i}\\
               \includegraphics[scale=0.7]{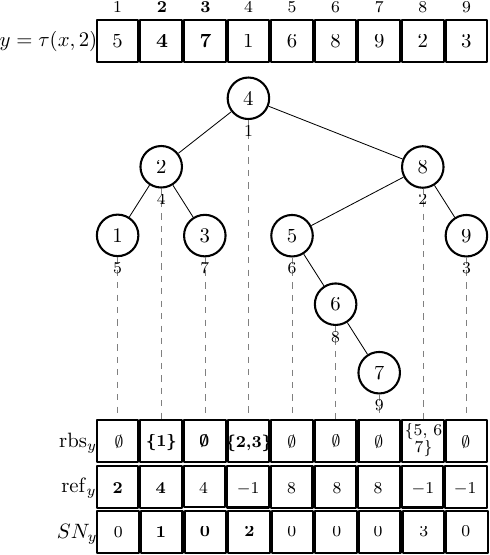} 
            \end{minipage}
            \begin{minipage}{0.49\textwidth}
                \ref{lem:proof:refi_neq_refi1_i_gti_1_ii}\\
                \includegraphics[scale=0.7]{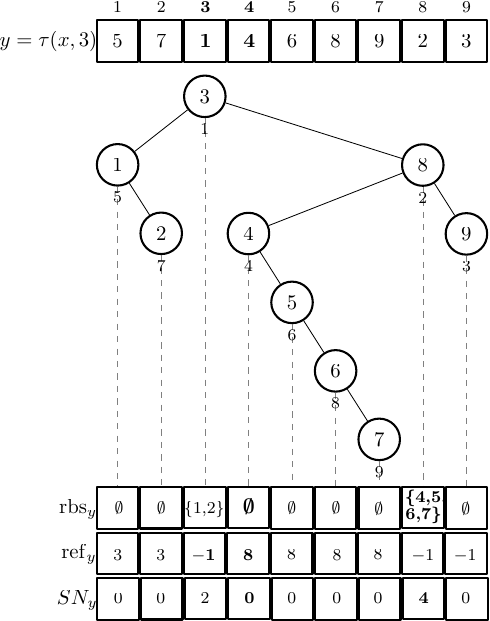} 
            \end{minipage}
            \caption{\label{fig:1}%
            From left to right, we show the effect of swaps on the sequence $x= (5,7,4,1,6,8,9,2,3)$ (from \Cref{fig:C-X}) at positions $5, 4, 2$ and $3$, which, see \Cref{lm:swaptomismatch,lm:changesLess,lm:changesMore}, respectively
             corresponds to cases 
             \ref{lem:proof:refi_eq_refi1}. ($\REFD{x}(i)\neq \REFD{y}(i)$), 
             \ref{lem:proof:refi_neq_refi1_i_lti_1}. ($\REFD{x}(i+1) \neq \REFD{y}(i+1)$  ),
             \ref{lem:proof:refi_neq_refi1_i_gti_1_i} ($\REFD{x}(i+1)=\REFD{y}(i) \neq \REFD{y}(i+1)$) and 
             \ref{lem:proof:refi_neq_refi1_i_gti_1_ii} ($\REFD{x}(i+1)\neq \REFD{y}(i+1)$).
             Values that are different from the ones in tables of
             \Cref{fig:C-X} are bolded.
             As one can see, $\SN_x$ and $\SN_y$ differ in at most 3 positions.
            }
            
\end{figure}

We now show that the \SNRepres\ of $x$ and $y$ at position $\REFD{y}(i)$
 cannot be different unless the \SNRepres of $x$ and $y$ at position $\REFD{x}(i+1)$ are different.

\begin{coro}\label{coro:swaptomismatch}
If $\REFD{y}(i)\ne -1 $ and $\REFD{x}(i+1)\ne -1$ then
 $\SN_x[\REFD{y}(i)] \neq \SN_y[\REFD{y}(i)]$ iff $\SN_x[\REFD{x}(i+1)] \neq \SN_y[\REFD{x}(i+1)]$
\end{coro}

\begin{proof}
We consider all cases of~\Cref{lm:swaptomismatch}.
\begin{enumerate}
    \item 
    Let $j = \REFD{x}(i) = \REFD{x}(i+1)$:
    Let us recall that $x[i] < x[i+1]$, $\REFD{y}(i) = i+1$, $\SN_x[i+1] = 0$, $\SN_y[i+1] \neq 0$ and that $\REFD{y}(i+1) = j$.
    Then, nodes $i$ and $i+1$ are skipped by node $j$ in $x$ while node $i$ is skipped by $i+1$ which is in turn skipped by $j$ in $y$, thus $\SN_x[j] = \SN_y[j] + 1$. 
    \item 
    $\REFD{x}(i) \neq \REFD{x}(i+1)$:
    \begin{enumerate}
        \item 
        $x[i] < x[i+1]$:
        Let us recall that $\REFD{y}(i) = i+1$,
        $\SN_x[i+1]=0$ and $\SN_y[i+1]\neq 0$.
        We have $\SN_x[\REFD{x}(i+1)] = \SN_y[\REFD{x}(i+1)] +1$ since node $i+1$ is not on the rightmost path of $C(y[1\ldots \REFD{x}(i+1)-1])$.
            
        \item 
        $x[i] > x[i+1]$:
        \begin{enumerate}
            \item 
            $\REFD{x}(i+1)\neq -1$:
            Let us recall that $j = \REFD{y}(i) = \REFD{x}(i+1)$.~\Cref{lm:swaptomismatch} shows that $\SN_x[j] = \SN_y[j]$.
            \item 
            $\REFD{x}(i+1)=-1$:
            Let us recall that we also have $\REFD{y}(i) = -1$ according to~\Cref{lm:swaptomismatch}.
        \end{enumerate}
    \end{enumerate}
\end{enumerate}
\end{proof}

However, having at most $3$ mismatches in the \SNRepres\ does not imply exactly one swap in the sequence, since more than one swap can affect the same positions provided above in \Cref{lm:swaptomismatch}.

\begin{example}
    Let $x=(8,7,6,5)$ and $\SN_x=(0,1,1,1)$ while $y=(7,8,5,6)$ and $\SN_y=(0,0,2,0)$. There are $3$ mismatches between $\SN_x$ and $\SN_y$, while there are two swaps between $x$ and $y$.
\end{example}

Thus, we propose a stronger lemma based on the four case analysis in the proof of \Cref{lm:swaptomismatch}, where we focus on the affected positions and show how exactly the \SNRepres\ differs at these positions when a swap occurs.

let $k$ denote the number of positions $1\leq j<i$ on the rightmost path of $C(x[1\ldots i-1])$ where $x[j]>x[i+1]$. 
Note that since we proved in \Cref{coro:swaptomismatch} that $\REFD{y}(i)$ cannot be different from $\REFD{x}(i)$ unless $\REFD{x}(i+1)$ is different from $\REFD{y}(i+1)$, thus we omit it in the next lemma. 

\begin{lemma}
    \label{lm:changesLess}
    If $x[i]<x[i+1]$, then 
    $$\begin{cases}\SN_{y}[i] \in \{0,\ldots, \SN_{x}[i]\} \\ \SN_{y}[i+1] = \SN_x[i]-\SN_{y}[i] + 1 \\ \SN_{y}[\REFD{x}(i+1)] = \SN_{x}[\REFD{x}(i+1)]-1, \text{ if } \REFD{x}(i+1) \neq -1 \end{cases}$$
\end{lemma}

\begin{proof}
In $C(y)$, node $i$ is swapped with node $i+1$
compared to $C(x)$
since $x[i]<x[i+1]$
implies that $y[i+1]<y[i]$.
Then node $i+1$ is removed from $\lb(\rightsb(C_i(x))$ and added somewhere in the right branch 
of $\leftsb(C_i(x))$ (see \figurename~\ref{fig:changeLess}). 
By definition, there is exactly $\SN_{x}[i]+1$ such positions. 
Once the node is inserted at a position $j\in \{0,\ldots, \SN_{x}[i]$, $\SN_{y}[i] = \SN_{x}[i] - j \in \{0,\ldots, \SN_{x}[i]\}$. 
Since we removed $\SN_{y}[i]$ nodes from 
the right branch of $\leftsb(C_i(x))$, but added former node $i+1$, we have $\SN_{y}[i+1] = \SN_x[i]-\SN_{y}[i] + 1$.
Finally, if $i+1$ has a
referent in $C(x)$, that is $\REFD{x}(i+1) \neq -1$, then both nodes $i$ and $i+1$ were in $\rb(\leftsb(\REFD{x}(i+1))$. Since we
removed node $i+1$ from this subtree, we have $\SN_{y}[\REFD{x}(i+1)] = \SN_{x}[\REFD{x}(i+1)]-1$.
\end{proof}

\tikzstyle{myState}=[draw=blue!50,very thick,fill=blue!20,circle,inner sep =1pt,minimum height=.75em]
\tikzstyle{myFinalState}=[draw=green!50,very thick,fill=green!20,circle,inner sep =1pt,minimum height=.75em]

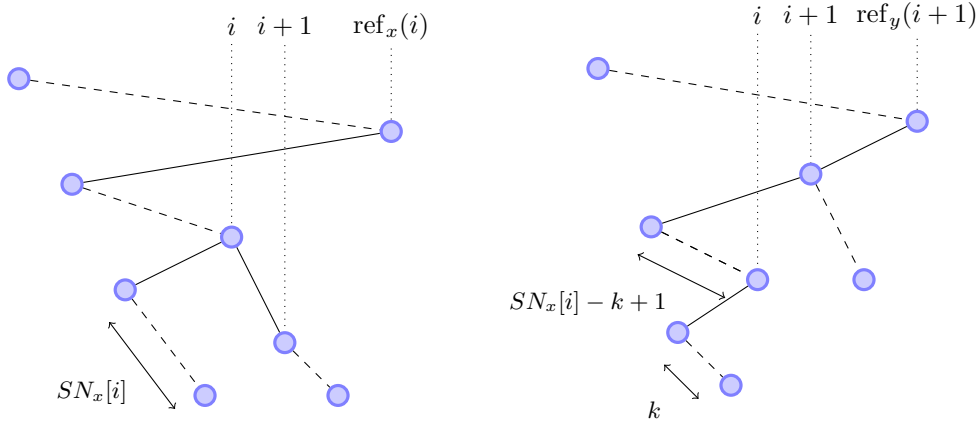
\begin{figure}
    \centering
\begin{tikzpicture}[scale=0.7]
\node at (4,7) (l1) {$i$};
\node at (5,7) (l2) {$i+1$};
\node at (7,7) (l3) {$\REFD{x}(i)$};

\node[myState] at (0,6) (q1) {}; 
\node[myState] at (1,4) (q2) {}; 
\node[myState] at (2,2) (q3) {}; 
\node[myState] at (3.5,0) (q4) {};
\node[myState] at (4,3) (q5) {}; 
\node[myState] at (5,1) (q6) {}; 
\node[myState] at (6,0) (q7) {}; 
\node[myState] at (7,5) (q8) {};
\path[-] (q3) edge node {} (q5)
         (q5) edge node {} (q6)
         (q2) edge node {} (q8)
;
\path[dashed] (q3) edge node {} (q4)
         (q1) edge node {} (q8)
         (q2) edge node {} (q5)
         (q6) edge node {} (q7)
;
\path[dotted] (l1) edge node {} (q5)
         (l2) edge node {} (q6)
         (l3) edge node {} (q8)
;
\node (a1) [below left=0.1cm of q3] {};
\node (a2) [below left=0.1cm of q4] {};
\path[<->] (a1) edge node [below left=0.1
cm] {\small $\SN_x[i]$} (a2);
\end{tikzpicture}
\qquad
\begin{tikzpicture}[scale=0.7]

\node at (3,7) (l1) {$i$};
\node at (4,7) (l2) {$i+1$};
\node at (6,7) (l3) {$\REFD{y}(i+1)$};
\node[myState] at (0,6) (q1) {}; 
\node[myState] at (1,3) (q2) {}; 
\node[myState] at (1.5,1) (q3) {}; 
\node[myState] at (2.5,0) (q4) {};
\node[myState] at (3,2) (q5) {}; 
\node[myState] at (4,4) (q6) {}; 
\node[myState] at (5,2) (q7) {}; 
\node[myState] at (6,5) (q8) {};
\path[-] (q3) edge node {} (q5)
         (q2) edge node {} (q6)
         (q6) edge node {} (q8)
;
\path[dashed] (q3) edge node {} (q4)
         (q2) edge node {} (q5)
         (q2) edge node {} (q5)
         (q1) edge node {} (q8)
         (q6) edge node {} (q7)
;
\path[dotted] (l1) edge node {} (q5)
         (l2) edge node {} (q6)
         (l3) edge node {} (q8)
;
\node (a1) [below left=0.1cm of q3] {};
\node (a2) [below left=0.1cm of q4] {};
\path[<->] (a1) edge node [below left=0.2cm] {\small $k$} (a2);
\node (a3) [below left=0.1cm of q2] {};
\node (a4) [below left=0.1cm of q5] {};
\path[<->] (a3) edge node [below left=0.1cm] {\small $\SN_x[i]-k+1$} (a4);
\end{tikzpicture}
    \caption{$C(x)$ (left) and $C(y)$ (right)}
    \label{fig:changeLess}
\end{figure}

\begin{lemma}
    \label{lm:changesMore}
    If $x[i] > x[i+1]$, then 
    $$\begin{cases} 
    \SN_{y}[i] = \SN_{x}[i] + \SN_{x}[i+1] - 1 \\
    \SN_{y}[i+1] = 0 \\
    \SN_{y}[pos] = \SN_{x}[pos] +1, \text{ where } pos \in \{\REFD{x}(i+1)\} \cup \lb(\rightsb(C_{i+1}(x))) 
    \end{cases}$$
\end{lemma}

\begin{proof}
This operation is the opposite of the one occurring in Lemma~\ref{lm:changesLess}, which 
explains why $\SN_{y}[i] = \SN_{x}[i] + \SN_{x}[i+1] - 1$. Since $x[i] > x[i+1]$, 
$\leftsb(C_{i+1}(y)$ is empty and $\SN_{y}[i+1] = 0$. Finally, node $i+1$ in $C(y)$
has a new referent $pos$, thus $\SN_{y}[pos]$ has to be incremented.
If nodes $i$ and $i+1$ have the same referent, then $pos=\REFD{x}(i+1)$. If not, then 
$pos$ is a node from the left-branch of $\rightsb(C_{i+1}(x))$. 
\end{proof}

Note that those lemmas prove that either 
$\SN_y[i]\ne \SN_x[i]$ or $\SN_y[i+1]\ne \SN_x[i+1]$.
Therefore, if we can detect less than $3$ mismatches between two sequences $\SN_x$ and $\SN_y$, then if there exists a position $i$  such that $y=\tau(x,i)$, it implies that the smallest position of the mismatches is either $i$ or $i+1$. 
Moreover, in the case where there are exactly $3$ mismatches,  the first two have to be in positions $i$ and $i+1$.
Henceforth, in the following section, where we describe our solution in detail, the verification step of the algorithm focuses at positions $i$ and $i+1$ of the mismatches as the potential positions of the sought swap.

\begin{algorithm}\label{algo:testequivSN}
  \caption{\algoName{equivalenceTestSwapSN}$((\SN_p), (\SN_x))$}
  \SetAlgoLined
  \SetKwInOut{KwIn}{Input}
  \SetKwInOut{KwOut}{Output}
  \KwIn{The \SNRepres tables of $p$ and $x$}
  \KwOut{$True$ if $p \tauCTM x$, $False$ otherwise}
  $j \leftarrow 2$\;
  \While{$j \le m$ and $\SN_p[j] = \SN_x[j]$ }{
  $j \leftarrow j+1$\;
  }
  \If(\tcp*[f]{Exact match}){$j=m+1$}{
        \Return $True$\; 
  }
  \If{$x[j]<x[j+1]$}{
     Check Lemma~\ref{lm:changesLess} and update $j$ accordingly\;
     $pos \leftarrow \REFD{x}(j+1)$\;
  }
  \Else{
     Check Lemma~\ref{lm:changesMore} and update $j$ accordingly\;
     $pos \leftarrow \REFD{y}(j+1)$\;
  }
  $k \leftarrow j+2$\;
  \While{$k \le m$}{
       \If{$k \neq pos$ and $\SN_p[k] \neq \SN_x[k]$}{\Return{$False$}\;}
     $k \leftarrow k+1$\;
  }
  \Return{$True$}\;
\end{algorithm}

\begin{lemma}
    Algorithm~\ref{algo:testequivSN} has a $\Theta(m)$ worst-case complexity, a $\Theta(1)$ best-case complexity.
\end{lemma}

\subsection{Updating the Skipped-number representation}

When searching for of a pattern $p$ of length $m$ in a text $t$
 of length $n$, the Skipped-number representation of $p$ is computed once in
 a preprocessing phase.
The searching phase looks at the text $t$ through a window of size $m$.
The Skipped-number representation of $t[1\ldots m]$
 is first computed, then for $m+1\le j \le n$ the Skipped-number representation
 of $t[j-m+1\ldots j]$ is computed by updating the Skipped-number representation
 of $t[j-m\ldots j-1]$.
For that, Function \algoName{updateSN} in Algorithm~\ref{algo:updateSN} uses:
\begin{itemize}
    \item
    a deque $D$
    for storing the right branch of the Cartesian tree of the current window on $t$;
    \item
    a circular array $SN$ of size $m$ for storing the Skipped-number representation
    of the current window on $t$;
    \item
    a circular array $RD$ of size $m$ for storing the distance from their referents of all the
    positions in the current window on $t$.
\end{itemize}

If the first element of $D$ is equal to $j-m$ then it is removed from $D$
 and the Skipped-number representation of its referent is decreased by one.
Then position $j$ can be inserted in the right branch (possibly popping some
 elements and updating the distance to their referent which is $j$)
 and its Skipped-number representation is computed.
The distance to its referent is set to $0$.
Since this update operation only adds a constant number of operations for each
 $m+1\le j \le n$ in addition to the computation of the Cartesian tree of
 $t$, it also has a linear worst case time complexity for all the calls
 to \algoName{updateSN}.

\begin{algorithm}
  \caption{\algoName{updateSN}$(t, j, m, S, SN, RD)$\label{algo:updateSN}}
  \SetAlgoLined
  \SetKwInOut{KwIn}{Input}
  \SetKwInOut{KwOut}{Output}
  \KwIn{%
        $t$: sequence, %
        $j$: position, %
        $m$: window size, %
        $D$: right branch of $C(t[j-m\ldots j-1]$, %
        $SN$: Skipped-number representation of $t[j-m\ldots j-1]$, %
        $RD$: distances to the referents of positions $[m-j\ldots j-1]$ %
        }
  \KwOut{$(D, SN, RD)$ such that %
         $D$: right branch of $C(t[j-m+1\ldots j]$, %
         $SN$: Skipped-number representation of $t[j-m+1\ldots j]$, %
         $RD$: distances to the referents of positions $[m-j+1\ldots j]$ %
         }
  $(val,pos) \leftarrow \dBack(D)$\;
  \If{$j-pos \ge m$}{
    $\dPopBack(D)$\;
  }
  \If{$RD[(j\bmod m)+1] > 0$}{
    $r \leftarrow ((j+RD[(j \bmod m)+1]) \bmod m) +1$\;
    $SN[r] \leftarrow SN[r] - 1$\;
  }
  $s \leftarrow 0$\;
  \While{not $\dIsEmpty(D)$}{
    $(val,pos) \leftarrow \dFront(D)$\;
    \If{$val < t[j]$}{
        Break\;
    }
    $RD[(pos \bmod m)+1] \leftarrow j-pos$\;
    $\dPopFront(D)$\;
    $s\leftarrow s+1$\;
  }
  $\dPushFront(D,(t[j],j))$\;
  $SN[(j\bmod m)+1]\leftarrow s$\;
  $RD[(j\bmod m)+1]\leftarrow 0$\;
  \Return{(D, SN, RD)}\;
\end{algorithm}

\section{Effect of one mismatch/insertion/deletion on linear representations} \label{sec:differences}
Given a pattern $p$ of length $m$ and a text $t$ of length $n$,
 the approximate Cartesian tree matching with at most one mismatch,
 one insertion or one deletion (see \Cref{ct-mis,ct-ins,ct-del}) can be solved by scanning the text
 with a sliding window of size $m$, $m+1$ or $m-1$ respectively.
 
 \begin{definition}[$\lcp, \lcs$]
Let $u$ and $v$ be two sequences.
Then let $\lcp(u,v)$ denote the length of the longest common
 prefix of $u$ and $v$ and
 let $\lcs(u,v)$ denote the length of the longest common
 suffix of $u$ and $v$.
 \end{definition}

\subsection{Effect of one mismatch on linear representations \label{sec:mismatch}}

\begin{lemma}\label{lm:mis}

Given $1\le j \le n-m+1$.
Let $\ell = \lcp(\PDD_p, \PDD_{t[j\ldots j+m-1]})$
 and
 let $r = \lcs(\PDG_p, \PDG_{t[j\ldots j+m-1]})$.
If $\ell+r \ge m-1$, then there exists an occurrence with at most one mismatch of $p$ in $t$.
\end{lemma}
\begin{proof}
If $\ell \ge m-1$ then $p[1\ldots m-1] \CTM t[j\ldots j+m-2]$ and thus
 $p \misCTM t[j\ldots j+m-1]$.
 
Otherwise, if $\ell < m-1$ and $\ell+r \ge m-1$
 then there exists $k\in\{1,\ldots, m\}$
 such that $p[1\ldots k-1] \CTM t[j\ldots j+k-2]$
 and $p[k+1\ldots m] \CTM y[j+k\ldots j+m-1]$
 and thus $p \misCTM t[j\ldots j+m-1]$.
\end{proof}

\begin{example}
Let
$p=(2,3,4,1,5,7,8,6,9)$,
$t=(4,3,
7,8,13,6,9,10,5,11
,1,2
)$
and let us consider the window on $t$:
$x=t[3\ldots 10]= (7,8,13,6,9,10,5,11)$

Then:

\noindent
$\PDD_p=(0, 1, 1, 0, 1, 1, 1, 3, 1)$ and
 $\PDG_p=(3, 2, 1, 0, 0, 2, 1, 0, 0)$,

\noindent
$\PDD_x=(0, 1, 1, 0, 1, 1, 1, 0, 1)$
and
$\PDG_x=(3, 2, 1, 4, 3, 2, 1, 0, 0)$.

Then
$\lcp(\PDD_p,\PDD_x) = 7$ and
$\lcs(\PDD_p,\PDD_x) = 4$ and $p[0\ldots 3] \CTM x[0\ldots 3]$ and $p[5\ldots 8] \CTM x[5\ldots 8]$
thus $p \misCTM x$.

See example \Cref{fig:CT-MIS-INS-DEL}(a).
\end{example}

\subsection{Effect of one insertion on linear representations \label{sec:insert}}


\begin{lemma} \label{lm:ins}
Given $1\le j \le n-m$.
Let $\ell = \lcp(\PDD_p, \PDD_{t[j\ldots j+m]})$
 and
 let $r = \lcs(\PDG_p, \PDG_{t[j\ldots j+m]})$.
If $\ell+r \ge m$, then there exists an occurrence with at most one insertion of $p$ in $t$.
\end{lemma}

\begin{proof}
If $\ell \ge m$ then $p[1\ldots m] \CTM t[j\ldots j+m-1]$ and thus
 $p \insCTM t[j\ldots j+m]$.
 
Otherwise, if $\ell < m$ and $\ell+r \ge m$
 then there exists $k\in\{1,\ldots, m\}$ such that $p[1\ldots k-1] \CTM t[j\ldots j+k-2]$
 and $p[k\ldots m] \CTM t[j+k\ldots j+m]$
 and thus $p \insCTM t[j\ldots j+m]$.
\end{proof}
 
 See example \Cref{fig:CT-MIS-INS-DEL}(b).

\subsection{Effect of one deletion on linear representations \label{sec:delete}}


\begin{lemma}\label{lm:del}
Given $1\le j \le n-m+2$.
Let $\ell = \lcp(\PDD_p, \PDD_{t[j\ldots j+m-2]})$
 and
 let $r = \lcs(\PDG_p, \PDG_{t[j\ldots j+m-2]})$.
If $\ell+r \ge m-1$, then there is an occurrence with at most one deletion of $p$ in $t$. 
\end{lemma}

\begin{proof}
If $\ell \ge m-1$ then $p[1\ldots m-1] \CTM t[j\ldots j+m-2]$ and thus
 $p \delCTM t[j\ldots j+m-2]$.
 
Otherwise, if $\ell < m-1$ and $\ell+r \ge m-1$
 then $\exists\ k\in\{1,\ldots, m\}$
 such that $p[1\ldots k-1] \CTM t[j\ldots j+k-2]$
 and $p[k+1\ldots m] \CTM t[j+k-1\ldots j+m-2]$
 and thus $p \delCTM t[j\ldots j+m-2]$.
\end{proof}

See example \Cref{fig:CT-MIS-INS-DEL}(c).

\subsection{An algorithm to test the equivalence \label{sec:testEQ}}

\begin{algorithm}\label{algo:testequivMID}
  \caption{\algoName{\eqTest Diff}$((\PDD_p, \PDG_p), (\PDD_x, \PDG_x))$}
  \SetAlgoLined
  \SetKwInOut{KwIn}{Input}
  \SetKwInOut{KwOut}{Output}
  \KwIn{The parent-distance tables of $p$ and $x$}
  \KwOut{$True$ if $x$ is equivalent to $p$, $False$ otherwise}
  $j \leftarrow 2$\;
  \While{$j \le m$ and $\PDD_p[j] = \PDD_x[j]$ }{
  $j \leftarrow j+1$\;
  }
  \If{$j=m+1$}{
        \Return $True$\;
  }
  $k = m - 1$\;
  \While{$k \ge j$ and $\PDG_p[k] = \PDG_x[k]$ }{
  $k \leftarrow k-1$\;
  }
  \If{Lemmas~\ref{lm:mis} (for $\misCTM$),~\ref{lm:ins} (for $\insCTM$), or~\ref{lm:del} (for $\delCTM$) hold for $p, x$ and $j$}{
    \Return{$True$}\;
  }
  \Return{$False$}\;
\end{algorithm}

Algorithm~\ref{algo:testequivMID} describes the \algoName{\eqTest} needed
in the \algoName{MetaAlgorithm} (Algorithm~\ref{algo:meta}) to solve the approximate Cartesian tree matching problem with up to one mismatch, insertion or deletion.
Note that the results from Section~\ref{sec:approx} on the worst-case
complexity and average-case complexity of the \algoName{MetaAlgorithm} still apply.

\section{Experiments \label{sec:expe}}
The following experiments were all performed using Python3 which had an influence on the obtained results.
The random model considered for the experiments is the uniform distribution
over permutations for the pattern and the text.
As one can see on \Cref{fig:av_pd_runtime,fig:av_pd_cmp,fig:av_diff_cmp,fig:av_diff_runtime},
the experimental results are consistent with \Cref{lm:avperm}.
The average running time decreases with the size of the pattern:
For each window starting position, the average number of comparisons tends to a constant smaller than $4$. The number of such positions the algorithms have to test decreases as the size of the pattern increases,  and so the total number of comparisons decreases. All experiments ran on a Dell Inc. Precision 3581 with a 13th Gen Intel Core™ i7-13700H × 20 CPU and 32 GB RAM.

\begin{figure}[!ht]
  \begin{minipage}[t]{0.47\textwidth}
  \vspace{0pt}
  \includegraphics[scale=0.5]{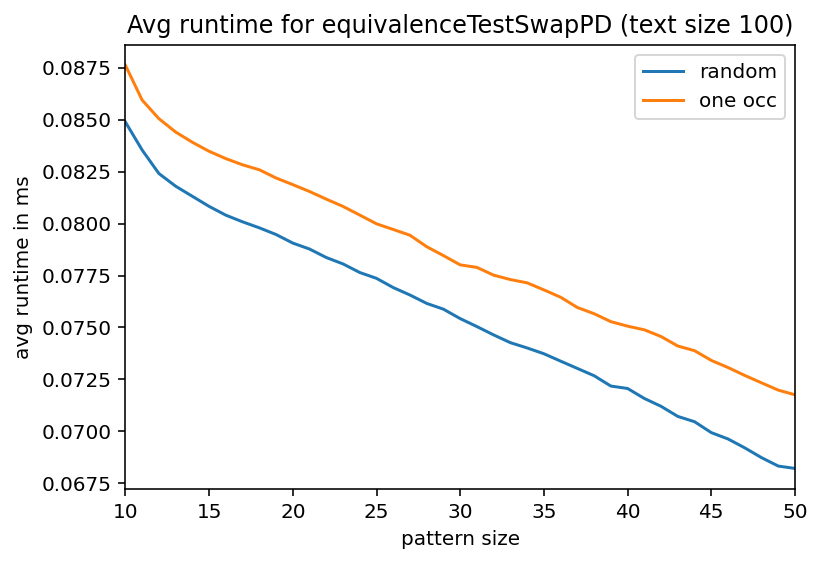}
  \end{minipage}
  \hfill
  \begin{minipage}[t]{0.47\textwidth}
  \vspace{0pt}
  \includegraphics[scale=0.5]{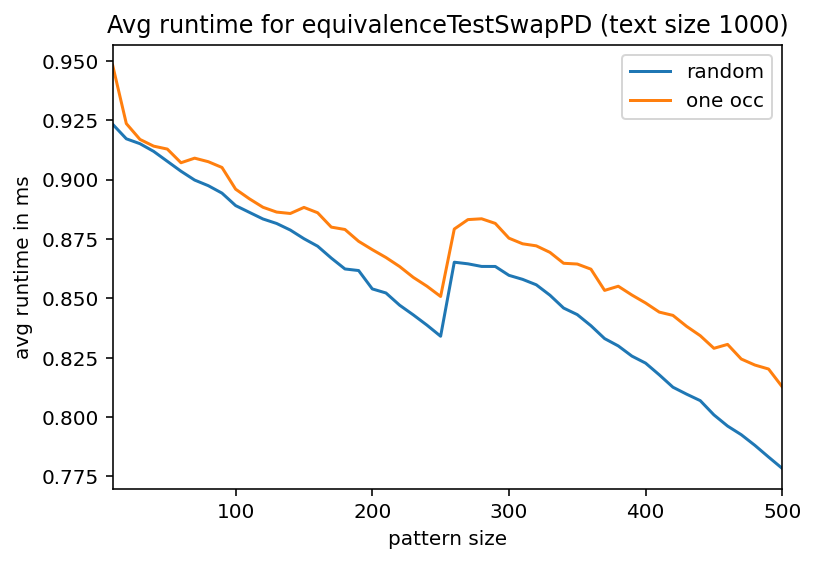}
  \end{minipage}
  \caption{Average runtime of the \algoName{MetaAlgorithm} using \algoName{equivalenceTestSwapPD}. 
  Uniform random permutations were generated $20\,000$ times for both the text and the pattern and a mean value was computed for each value of each curves.
  The bumps that occurs in the second graph is due to the fact that, when the values are above $256$, Python changes the type of the variables, with an additional cost. The blue curve represents fully random data while at least one occurrence of the pattern was guaranteed in the text for the orange curve. \label{fig:av_pd_runtime} }
\end{figure}

\begin{figure}[!ht]
  \begin{minipage}[t]{0.49\textwidth}
  \vspace{0pt}
  \includegraphics[scale=0.5]{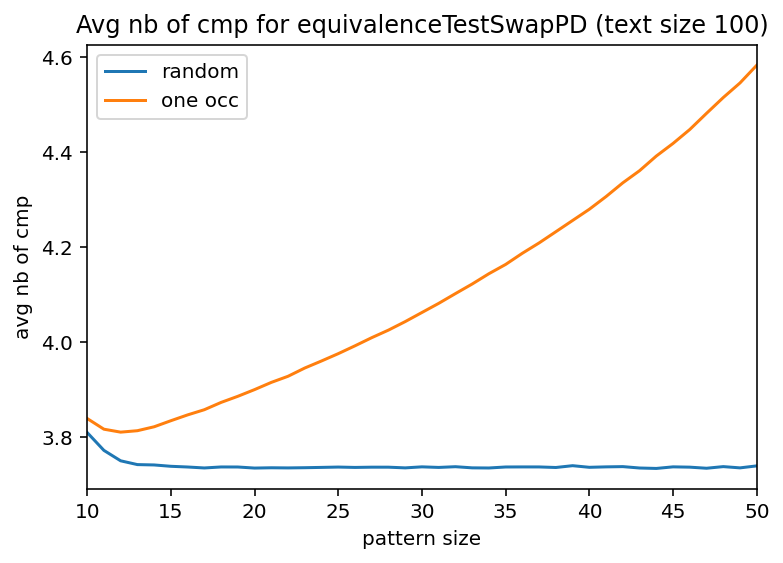}
  \end{minipage}
  \hfill
  \begin{minipage}[t]{0.49\textwidth}
  \vspace{0pt}
  \includegraphics[scale=0.5]{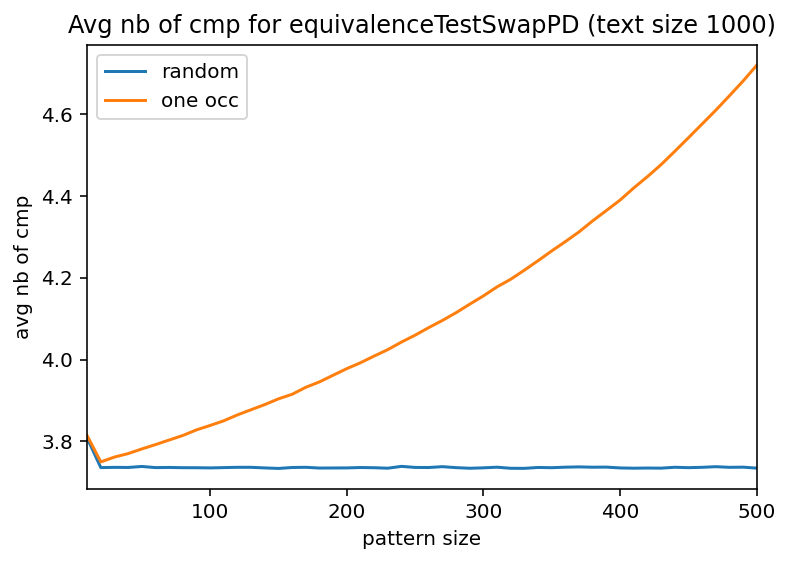}
  \end{minipage}
  \caption{\label{fig:av_pd_cmp} Average number of comparisons per window starting position performed by the \algoName{equivalenceTestSwapPD} algorithm. As one can see with the blue curve, the average number of comparisons tends to a constant when it is unlikely to find occurrences of the pattern. The slight bump at the beginning of the curve is due to smaller patterns being scarcely found, hence requiring more comparisons to verify the red and blue zones. As demonstrated by the orange curve, the more occurrences found (or the larger the pattern, in comparison to the text), the more expensive the algorithm becomes. Eventually, finding an occurrence of the pattern at every position of the text will lead to a quadratic cost.}
\end{figure}

\begin{figure}[!ht]
  \begin{minipage}[t]{0.47\textwidth}
  \vspace{0pt}
  \includegraphics[scale=0.5]{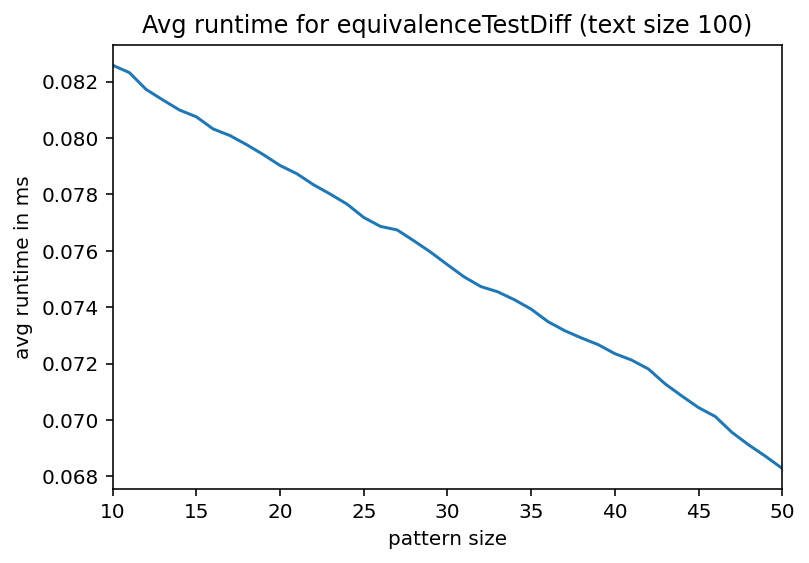}
  \end{minipage}
  \hfill
  \begin{minipage}[t]{0.47\textwidth}
  \vspace{0pt}
  \includegraphics[scale=0.5]{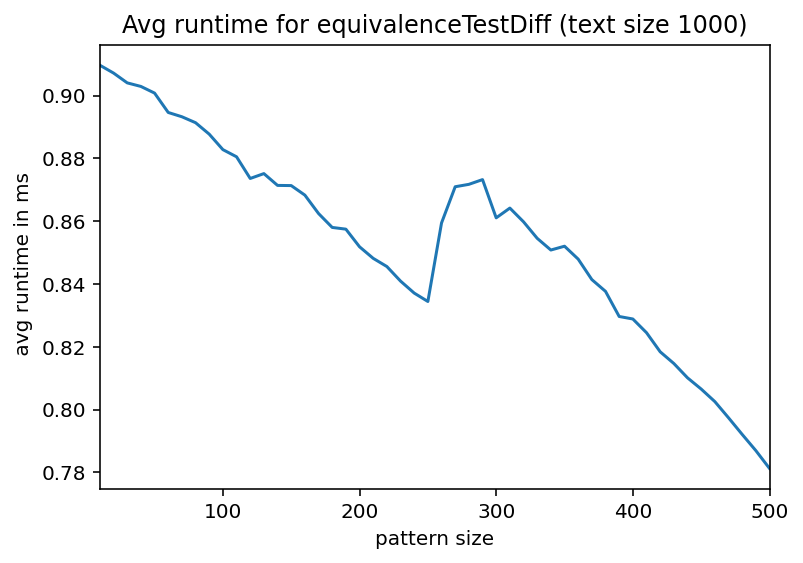}
  \end{minipage}
  \caption{Average run time of the \algoName{equivalenceTestDiff} algorithm. The above graphs showcase the results for approximate Cartesian tree matching with up to one mismatch, but insertion and deletion yield similar results.
  \label{fig:av_diff_runtime}}
\end{figure}

\begin{figure}[!ht]
  \begin{minipage}[t]{0.49\textwidth}
  \vspace{0pt}
  \includegraphics[scale=0.5]{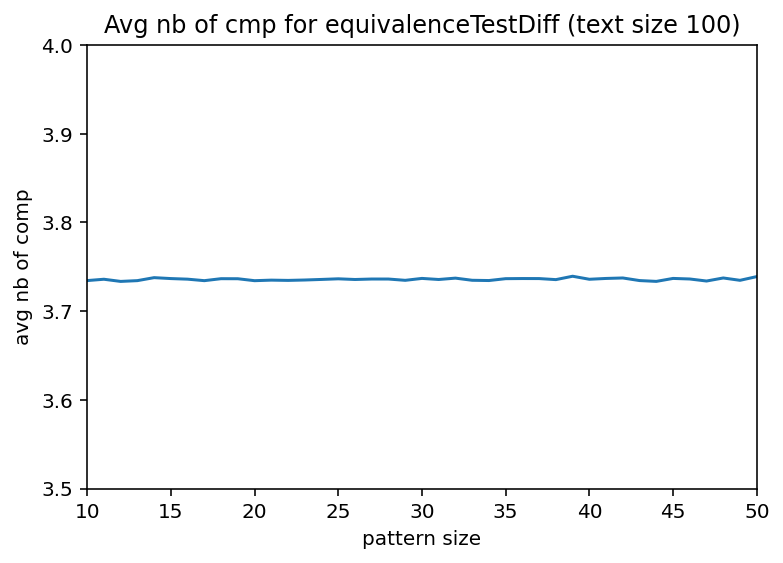}
  \end{minipage}
  \hfill
  \begin{minipage}[t]{0.49\textwidth}
  \vspace{0pt}
  \includegraphics[scale=0.5]{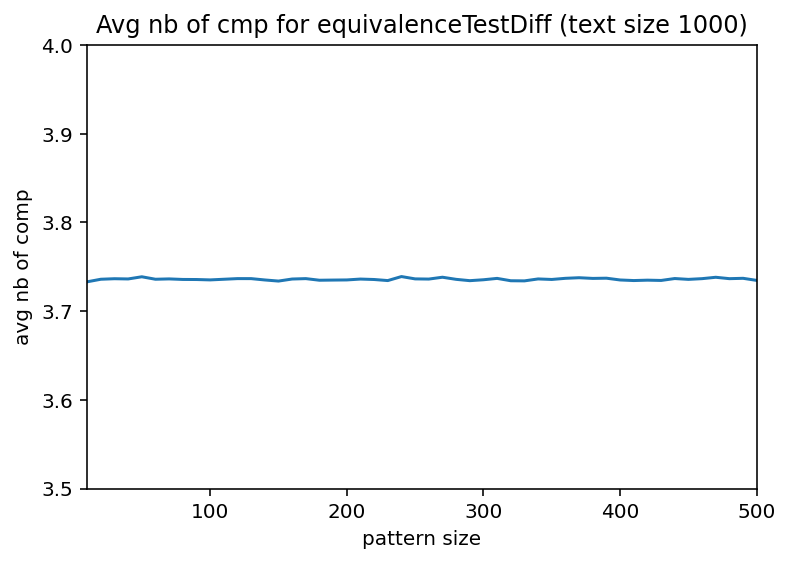}
  \end{minipage}
  \caption{Average number of comparisons per window starting position performed by the \algoName{equivalenceTestDiff} algorithm. The above graphs showcase the results for approximate Cartesian tree matching with up to one mismatch, but insertion and deletion yield similar results. Once again, one can see the average number of comparisons tends to a constant.
  \label{fig:av_diff_cmp}}
\end{figure}

\section{Perspectives \label{sec:persp}}
From the pattern matching point of view, the first step would be to generalize
our result to sequences with a partial order instead of a total one.

The question of whether the Aho-Corasick method could be adapted to errors 
like a mismatch, insertion or deletion is also still open.

Then, it could be interesting to obtain a general method, where the number
of swaps is given as a parameter. Though, we fear that if too many
swaps are applied, the result loses its interest, even though the complexity 
might grow rapidly.

\section*{Acknowledgements \label{sec:acknow}}
We would like to thank Julien Courtiel for his comments on the average analysis of the MetaAlgorithm.
We also thank Simone Faro for fruitful discussions on approximate Cartesian tree matching with one 
 mismatch, one insertion and one deletion.
B. Auvray, J. David, R. Groult and T. Lecroq were supported by the CNRS NormaSTIC federation.

\bibliographystyle{splncs04}
\bibliography{biblio}



\end{document}